\documentclass[11pt,a4paper]{article}

\usepackage{amsmath,amssymb,amsthm}
\usepackage{graphicx}
\usepackage[margin=2.5cm]{geometry}
\usepackage{times}
\usepackage{natbib}

\usepackage{moreverb}

\DeclareMathOperator{\Var}{Var}
\DeclareMathOperator{\Cov}{Cov}
\newcommand{\abs}[1]{\left\vert#1\right\vert}

\usepackage{caption}
\usepackage{subcaption}
\usepackage{wrapfig}
\usepackage{amsmath}
\usepackage{boxhandler} 
\usepackage{multicol} 
\usepackage{epstopdf}
\usepackage{microtype}

\usepackage[normalem]{ulem}

\usepackage{graphicx} 

\usepackage[colorlinks,bookmarksopen,bookmarksnumbered,citecolor=red,urlcolor=red]{hyperref}

\newcommand\BibTeX{{\rmfamily B\kern-.05em \textsc{i\kern-.025em b}\kern-.08em
T\kern-.1667em\lower.7ex\hbox{E}\kern-.125emX}}

\newcommand{\E}{\mathbb{E}}
\newcommand{\R}{\mathbb{R}}

\newcommand{\set}[1]{\left\{#1\right\}}

\newtheorem{thm}{Theorem}
\newtheorem{prop}[thm]{Proposition}
\newtheorem{assumption}[thm]{Assumption}
\newtheorem{prediction}[thm]{Prediction}

\begin{document}

\title{Properties of the Affine Invariant Ensemble Sampler's
`stretch move' in high dimensions}

\author{David Huijser\thanks{Department of Statistics, The University of Auckland, Private Bag 92019, Auckland 1142, New Zealand {\tt dhui890@aucklanduni.ac.nz}, {\tt jesse.goodman@auckland.ac.nz}, {\tt brendon.brewer@gmail.com}}~\footnote{To whom correspondence should be addressed} \and Jesse Goodman\footnotemark[1] \and Brendon J. Brewer\footnotemark[1]}


\maketitle

\begin{abstract}
We present theoretical and practical properties of the affine-invariant
ensemble sampler Markov Chain Monte Carlo method. 
In high dimensions, the sampler's `stretch move'
has unusual and undesirable properties. 
We demonstrate this with an $n$-dimensional correlated Gaussian toy problem with a known mean and covariance structure,
and a multivariate version of the Rosenbrock problem. 
Visual inspection of trace plots suggests the burn-in period is short.
Upon closer inspection, we discover the mean and the variance of the target distribution do not match the known values, and the chain takes a very long time to
converge.
This problem becomes severe as $n$ increases beyond 50. 
We therefore conclude that the stretch move should not
be relied upon (in isolation) in moderate to high dimensions. 
We also present some theoretical results explaining this behaviour.

\textsl{Key words:} Affine Invariant Ensemble Sampler, Stretch Move, Markov Chain Monte Carlo
\end{abstract}


\maketitle

\section{Introduction}
Since the introduction of the Markov Chain Monte Carlo methods (MCMC)
\citep{MetRosTel1953}, a large number of different algorithms have been developed.
Popular examples include Metropolis-Hastings \citep{Has1970},
slice sampling \citep{Nea2003} and Hamiltonian MCMC \citep{Neal2011}. 
Each has its own strengths and weaknesses.
A recent innovative MCMC method is the affine-invariant ensemble sampler (AIES) introduced by \cite{GooWea2009}. Methods that are invariant under
affine transformation of the parameter space offer much promise
for highly dependent target distributions, because there is usually
{\em some} affine transformation that would make the target density much easier
to sample from,
and the sampler performs identically on the untransformed problem
as it would on the transformed one.

The intuitition behind the `stretch move' of the AIES (described in Section~\ref{sec:aies}) is compelling. 
It is also straightforward to implement because the user doesn't need
to define a proposal distribution, or add any additional information except the ability to evaluate a function proportional to the density of the target distribution.  
This allowed \citet{ForHogLanGoo2013} to develop a high quality Python software implementation \verb"emcee", where the user only needs to implement a function that evaluates the unnormalized
target density. 
Performance comparisons \citep[e.g.][]{Lampart2012} show that the AIES is competitive with other common techniques, and outperforms them on certain kinds of target distribution.
As a result, the algorithm has become popular, especially in astronomy \citep{VanMonJoh2015,CroPetSch2015}.

However, questions remain about the behaviour of the AIES on high dimensional problems. 
In Section~\ref{sec:gaussian} we test the method on a correlated Gaussian target distribution and discuss the observed behaviour that arises in high dimensions. In Section~\ref{sec:rosenbrock}
we see the same problem on a more elaborate example.
Section~\ref{sec:taies} consists of a mathematical exploration of this behaviour.

\section{The AIES algorithm}\label{sec:aies}

The AIES algorithm works by evolving a set of $L$ samples, called {\it walkers}, of $n$ parameters.
A walker can be considered as a vector in the $n$-dimensional parameter space.
One iteration of AIES involves a sweep over all $L$ walkers.
For each walker, a new position is proposed, and
accepted with a
Metropolis-Hastings type acceptance probability.
The aim is to simulate the distribution specified by a target density function
$\pi(\mathbf{x})$ on the $n$-dimensional parameter space,
but, as is common for ensemble methods, the target distribution is
actually
\begin{align}
\prod_{i=1}^L \pi(\mathbf{x}_i),
\end{align}
that is, the target distribution $\pi$
independently replicated $L$ times, once for each walker.

Several kinds of proposals are possible, and we describe the \emph{stretch move} used in \verb"emcee".
Superscripts will denote walkers and subscripts will denote coordinates.
Thus $\mathbf{X}^{(j)}(t)$ means the position (in $n$-dimensional space) of the $j^\text{th}$ walker, $j=1,\dotsc,L$, at discrete time $t$ during the algorithm, and $X^{(j)}_i(t)$ means the $i^\text{th}$ coordinate of that walker, $i=1,\dotsc,n$.

At each iteration $t$, each walker is updated in sequence.
To update the $k$th walker, we select
select a complementary walker $\mathbf{Y}=\mathbf{Y}(t)=\mathbf{X}^{(j)}(t)$ with $j\neq k$ chosen uniformly, and define the proposal point
 \begin{equation}
 \widetilde{\mathbf{X}} =  Z\mathbf{X}^{(k)}   + (1-Z)\mathbf{Y}
 \label{eq:Proposal}
 \end{equation}
where $Z$ is a real-valued stretching variable drawn according to the density
\begin{equation}
g(z) \propto
\begin{cases}
\frac{1}{\sqrt{z} } & \text{ if } z \in \left[\frac{1}{a},a \right]  \\
0                   & \text{otherwise} 
\end{cases}
 \label{eq:Zdensity}
\end{equation}
where $a$ is an adjustable parameter, usually set to $2$ which is considered a good value in essentially all situations \citep{ForHogLanGoo2013}.
Finally, the proposal $\widetilde{\mathbf{X}}$ is accepted
to replace $\mathbf{X}^{(k)}$ with probability
 \begin{equation}
\begin{aligned}
p(\mathbf{X},\mathbf{Y},Z)
&=\min\left( 1, Z^{n-1} \frac{\pi\bigl(\widetilde{\mathbf{X}} \bigr) }{ \pi\left( \mathbf{X} \right)}\right).
\end{aligned}
 \label{eq:AcceptanceProb}
 \end{equation}
Otherwise, $\mathbf{X}^{(k)}$ remains unchanged.

In words, one can imagine the two selected walkers, $\mathbf{X}$
(the one that might be moved)
and $\mathbf{Y}$ (the one that helps construct the proposal),
defining a line in parameter space. 
The proposal is to move the main walker $\mathbf{X}$ to a new position along the line connecting $\mathbf{X}$ to $\mathbf{Y}$.
The stretching variable $Z$ defines how far the main walker moves along this line (either towards or away from $\mathbf{Y}$) to obtain a proposed new position, with $Z=1$ corresponding to no change.
Similar to the single particle Metropolis-Hastings sampler, the acceptance probability depends on the ratio of the target densities at the current and proposal points, with an additional factor $Z^{n-1}$ arising because the proposed position is chosen from a one dimensional subset of the $n$-dimensional space.

To enable parallel processing, the implementation in {\tt emcee} performs several stretch moves simultaneously.
The vector of walkers is split into two subsets $S^{(0)} = \{ \mathbf{X}^{(k)}\colon  k=1, \dots, L/2\}$ and $S^{(1)} = \{ \mathbf{X}^{(k)}\colon k=L/2, \dots, L\}$.
In the first ``half-iteration'', all walkers in $S^{(0)}$ are simultaneously updated according to the stretch move described above, with all the complementary walkers chosen from $S^{(1)}$.
In the second half-iteration, the sets are switched and all walkers in $S^{(1)}$ are simultaneously updated, with complementary walkers chosen from $S^{(0)}$.
The time variable $t$ denotes the number of iterations, each consisting of a pair of half-iterations. Because of subtleties related to detailed balance, the main and complementary walkers must not be updated simultaneously, hence the splitting into two subsets $S^{(0)},S^{(1)}$.

Later, we will also consider a simpler continuous-time variant without the subsets $S^{(0)},S^{(1)}$.
In this setup, at times $t$ chosen according to an independent exponential clock, a main walker $\mathbf{X}^{(k)}(t)$ and a complementary walker $\mathbf{X}^{(j)}(t)$ are chosen uniformly among all walkers, and a stretch move is performed.
To ensure consistency of the time variable, we set the overall rate of moves to equal $L$, the number of walkers, so that each walker is chosen as main walker once per unit of time, on average.

The walkers collectively -- i.e., the vector $\mathbf{X}(t)=\left(\mathbf{X}^{(1)}(t), \dots, \mathbf{X}^{(L)}(t)\right)$ in $n\times L$-dimensional space -- form a Markov chain under either of these dynamics (either in discrete time or in continuous time, respectively).
Properties of the scaling variable $Z$ and the acceptance probability in \eqref{eq:AcceptanceProb} ensure that this Markov chain has an equilibrium distribution corresponding to the target density $\pi$.
Specifically, the equilibrium distribution is that the walkers $\mathbf{X}^{(j)}$, $j=1,\dotsc,L$, are independent random samples from the density $\pi$.
Under mild conditions, this is the unique equilibrium distribution and any initial distribution will approach it as $t\to\infty$, provided that the initial points $\mathbf{X}^{(j)}(0)$ do not lie in an $(n-1)$-dimensional affine subspace of the parameter space.
In particular, this requires that $L\geq n+1$ in the non-parallel version.
Hence, for each sufficiently large time $t$, \emph{empirical means} such as 
  \begin{equation}
    \frac{1}{L}\sum_{j=1}^L f\left(\mathbf{X}^{(j)}(t)\right)
  \end{equation}
can be used as approximations of the integral $\int_{\mathbb{R}^n} f(\mathbf{x}) \pi(\mathbf{x}) d\mathbf{x}$, corresponding to the mean $\E\left(f(\mathbf{X})\right)$ when $\mathbf{X}$ is a random variable with density function $\pi$.
Similarly, empirical variances can be used to approximate $\Var\left(f\left(\mathbf{X}\right)\right)$.
These empirical means and variances can also be averaged over different values of $t$.
If this averaging starts after the Markov chain has burned in and spans a sufficient time compared to the mixing time, the overall estimate will improve.

As its name suggests, the AIES is invariant under affine transformations of parameter space.
To explain this property, suppose $\mathbf{X}$ has density $\pi$.
Given an invertible $n\times n$ matrix $\mathbf{A}$ and $n$-dimensional vector $\mathbf{b}$, define the affine transformation $\mathbf{x} \mapsto \mathbf{A}\mathbf{x}+\mathbf{b}$ and the random variable $\mathbf{Q}=\mathbf{A}\mathbf{x}+\mathbf{b}$.
Then $\mathbf{Q}$ has density 
\begin{equation}
\pi'(\mathbf{q}) = \frac{\pi\left(\mathbf{A}^{-1}(\mathbf{q} - \mathbf{b})\right)}{\det \mathbf{A}}.
\label{eq:Qdensity}
\end{equation}
The fact that the proposal point in \eqref{eq:Proposal} is a linear combination of existing walkers causes the AIES algorithm to be invariant under affine transformations:

\begin{prop}[Affine invariance property]\label{prop:AffineInvariance}
  Running the AIES with initial conditions $\mathbf{X}^{(j)}(0)$ and density $\pi$ is equivalent to running the AIES with initial conditions $\mathbf{Q}^{(j)}(0)=\mathbf{A}\mathbf{X}^{(j)}(0)+\mathbf{b}$ and density $\pi'$.
\end{prop}

In particular, the AIES algorithm does not give special treatment to moves along the coordinate axes.
 
\section{The AIES for sampling a high dimensional Gaussian}\label{sec:gaussian}

The AIES has been used with great success in various research projects \citep{VanMonJoh2015,CroPetSch2015}, and is especially popular in the astronomy community. 
However there is reason for caution if one tries to apply this method in higher dimensional problems $(n > 50)$, as we will show. 
Unfortunately, the output from AIES may resemble the output of an MCMC algorithm
``in equilibrium'', yet the points obtained from the AIES might not accurately represent the target distribution, with the true equilibrium taking much longer to achieve.

To investigate the properties of the AIES in $n$ dimensions we chose a correlated $n$-dimensional Gaussian as the target distribution
(See also the correlated Gaussian studied by \citet{Lampart2012}).
More specifically, the target distribution is a discrete-time Ornstein-Uhlenbeck process, also known as a discrete-time autoregressive process of the order 1, hereafter referred to as an AR(1) process. 
This model is well suited to be used for benchmarking, because
posterior distributions in Bayesian statistics are often approximately
multivariate normal. 
Besides this, the AR(1) is also useful as a prior in time series modelling.

The AR(1) distribution is the distribution of the random vector $\mathbf{X}$ whose coordinates are defined recursively by   
\begin{equation}
\begin{aligned}
X_1 &\sim N(0,1) \\
X_2|X_1 &\sim N(\alpha X_1, \beta^2) \\
X_3|X_2 &\sim N(\alpha X_2, \beta^2) \\
& \vdots \\
X_n|X_{n-1} &\sim N(\alpha X_{n-1}, \beta^2) 
\end{aligned}
\label{eq:AR1distr}
\end{equation}
where $N(\mu,\sigma^2)$ denotes a normal distribution, and
$\alpha$ controls the degree of correlation from one coordinate to the
next. We set $\beta = \sqrt{(1 - \alpha^2 )}$ so the marginal distribution
of all of the coordinates is $N(0, 1)$. 
If we run MCMC to sample this target distribution,
it should be straightforward to verify whether the output is correct, since
the expected values and standard deviations of all coordinates are 0 and 1
respectively.
To test the AIES,
we arbitrarily chose the coordinate $x_1$ as a probe of the convergence
properties of the AIES.

We sampled the AR(1) target distribution using \verb"emcee" \citep{ForHogLanGoo2013}
with $\alpha$ set to 0.9. We tested three values of
the dimensionality: $n=10$, $n=50$, and $n=100$, and set the number of walkers
$L$ to $2n$ in each case.
Each run consisted of 200,000 iterations (each of which is a loop over all
walkers). Each run was thinned to reduce the size of the
output.

For reasons explained in the next section, for each value of the dimensionality
$n$, we performed four separate runs, each of which had the starting positions sampled from four different widely dispersed distributions. These distributions are $N(0,5^2), N(1,5^2),N(-1,5^2)$, and $N(1,10^2)$. The initial conditions were generated by drawing each coordinate of each walker independently from these distributions. \\

\indent For a properly working MCMC method applied to this problem, the output should have an observed mean $\hat{\mu}\approx 0$ and an observed standard deviation $\hat{\sigma} \approx 1$. However, the observed values of $\hat{\sigma}$ for the obtained target distribution for $n=100$ dimensions (displayed in Table~\ref{tab:gaussian_mean}) are smaller than the true value $\sigma=1$. The results for $n=50$ also
appear to be suspect but to a lesser degree.

\begin{figure}[p]
\begin{subfigure}{0.31\textwidth}
\includegraphics[width=0.85\linewidth]{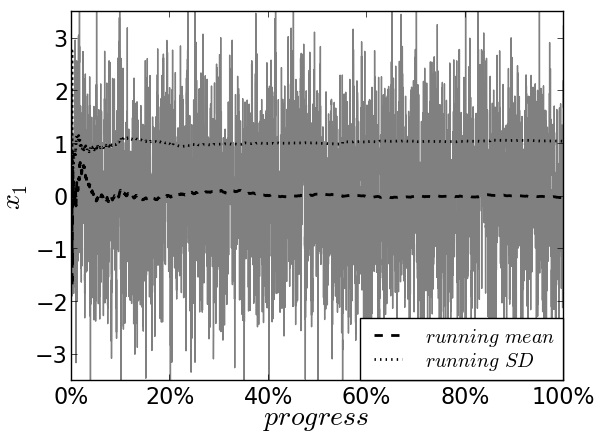}
\caption{$n=10$}
\label{fig:1a}
\end{subfigure}
\begin{subfigure}{0.31\textwidth}
\includegraphics[width=0.85\linewidth]{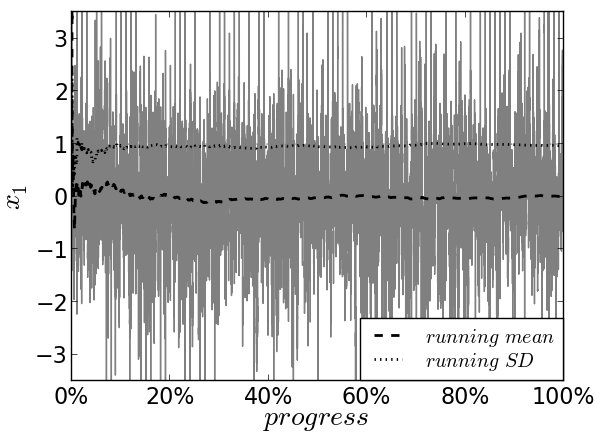}
\caption{$n=50$}
\label{fig:1b}
\end{subfigure}
\begin{subfigure}{0.31\textwidth}
\includegraphics[width=0.85\linewidth]{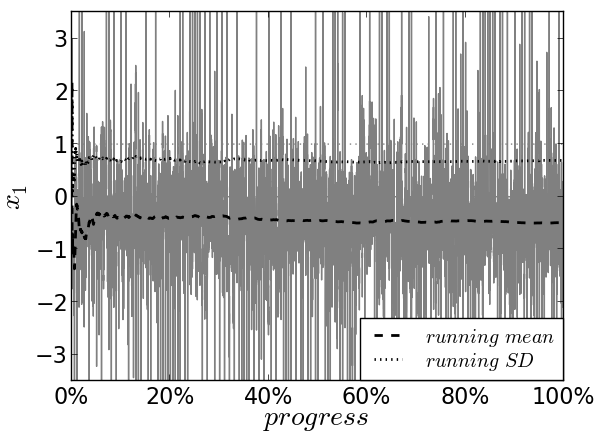}
\caption{$n=100$}
\label{fig:1c}
\end{subfigure}
\\
\vspace{-10pt}
\caption*{Graphs of the flattened trace plots of the first coordinate $x_1$ for $n=10$, $n=50$ and $n=100$. The $x$-axis is proportional to CPU time. The running
means and standard deviations are averaged over the second half of the
run, so progressively exclude more of the initial part of the run as time increases.
The $n=10$ and $n=50$ runs give more or less accurate results by the end of the run,
but $\hat{\sigma}$ is too small in the   $n=100$ run even though the trace plot might look satisfactory to the eye.}
\vspace{5pt}

\begin{subfigure}{0.31\textwidth}
\includegraphics[width=0.85\linewidth]{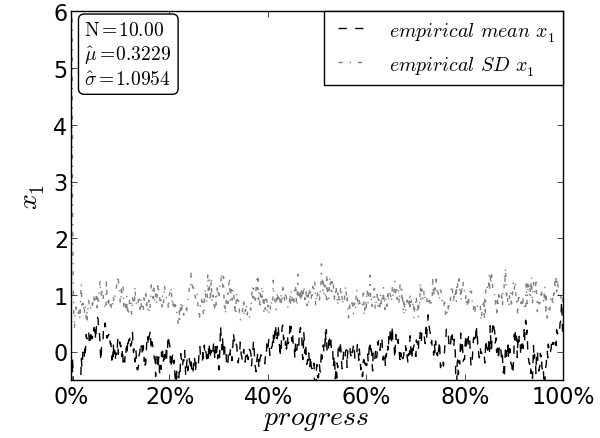}
\caption{$n=10$}
\label{fig:1d}
\end{subfigure}  
\vspace{3pt}
\begin{subfigure}{0.31\textwidth}
\includegraphics[width=0.85\linewidth]{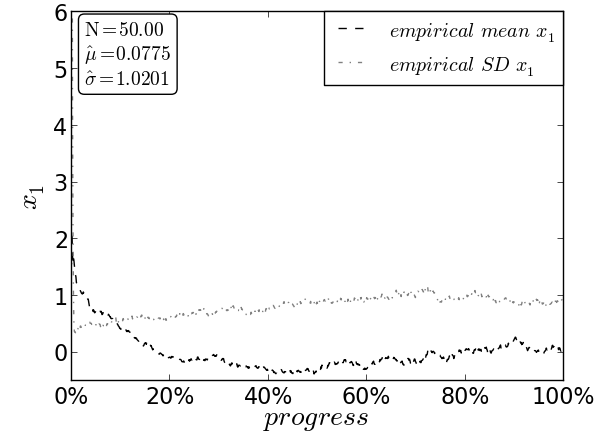}
\caption{$n=50$}
\label{fig:1e}
\end{subfigure}  
\vspace{3pt}
\begin{subfigure}{0.31\textwidth}
\includegraphics[width=0.85\linewidth]{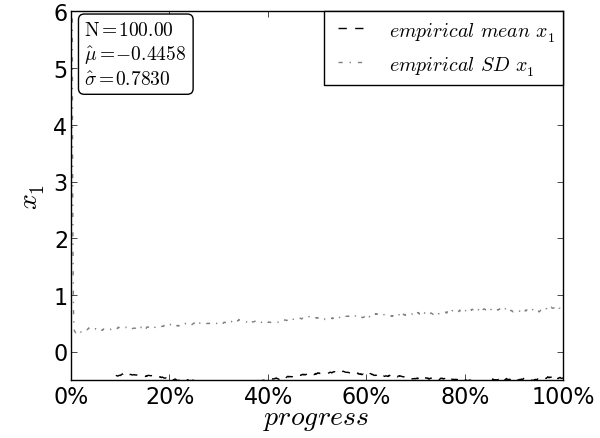}
\caption{$n=100$}
\label{fig:1f}
\end{subfigure} 
 \\
\vspace{-10pt}
\caption*{Mean values (grey) and variance (black) of $x_1$,
averaged over all walkers at a fixed time,
as a function of thinned iteration $t$.}
\vspace{5pt}

\begin{subfigure}{0.31\textwidth}
\vspace{0pt}
\includegraphics[width=0.99\linewidth]{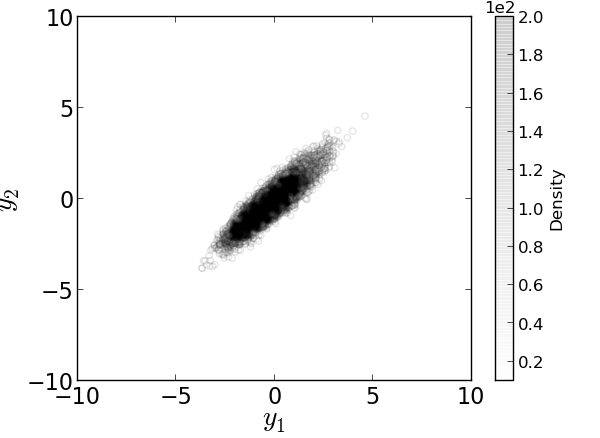}
\caption{$n=10$}
\label{fig:1g}
\end{subfigure} 
\begin{subfigure}{0.31\textwidth}
\vspace{0pt}
\includegraphics[width=0.99\linewidth]{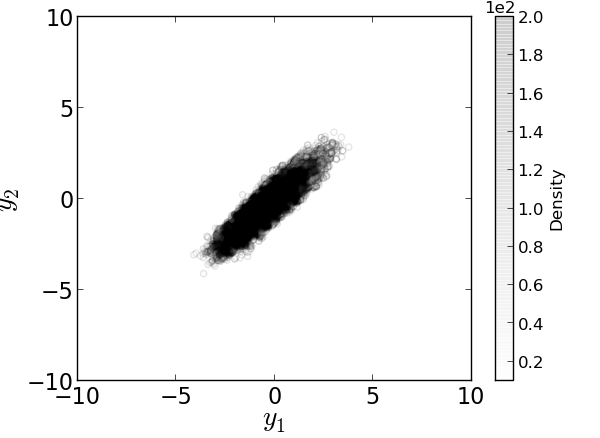}
\caption{$n=50$}
\label{fig:1h}
\end{subfigure} 
\begin{subfigure}{0.31\textwidth}
\vspace{0pt}
\includegraphics[width=0.99\linewidth]{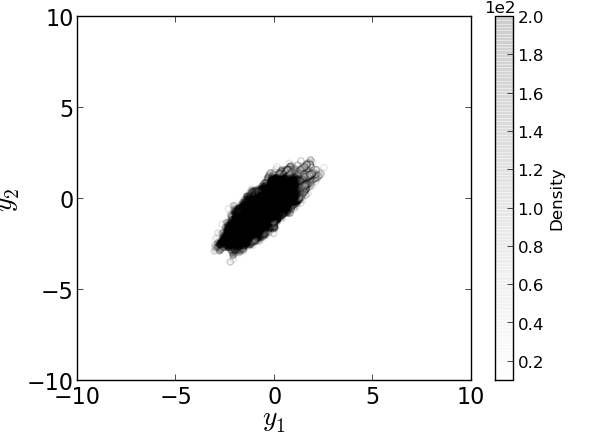}
\caption{$n=100$}
\label{fig:1i}
\end{subfigure}  
\\
\vspace{-10pt}
\caption*{Scatter plot of binned $y_1$ vs. $y_2$, which are defined as 
the average taken over all walkers for parameters $x_1$ and $x_2$ for 
the second half of the run. The marginal distribution for $y_1$ and $_2$
should resemble the left plot here. However, for the $n=100$ run,
the correlation is incorrect.}
\vspace{5pt}

\begin{subfigure}{0.31\textwidth}
\includegraphics[width=0.99\linewidth]{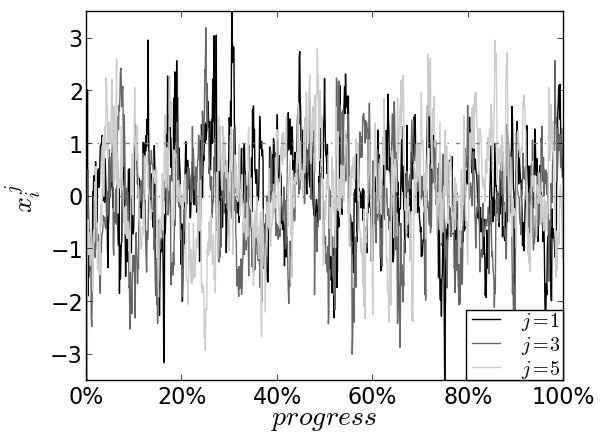}
\caption{$n=10$} 
\label{fig:1j}
\end{subfigure}
\begin{subfigure}{0.31\textwidth}
\includegraphics[width=0.99\linewidth]{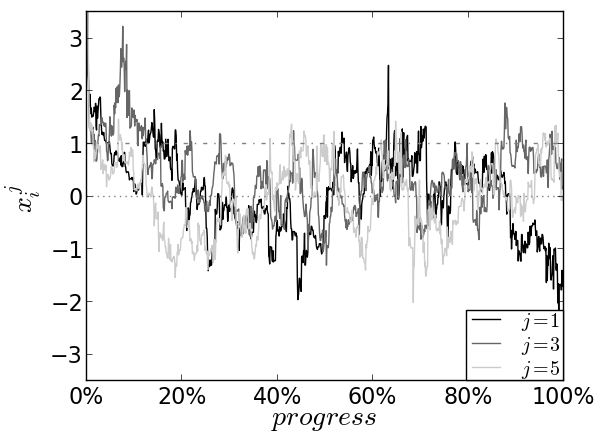}
\caption{$n=50$} 
\label{fig:1k}
\end{subfigure}  
\begin{subfigure}{0.31\textwidth}
\includegraphics[width=0.99\linewidth]{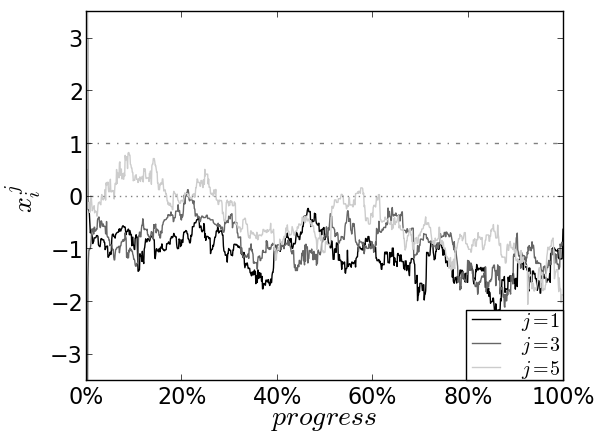}
\caption{$n=100$} 
\label{fig:1l}
\end{subfigure} 
\vspace{-10pt}
\caption*{Trace plot of the coordinate $x_1$ for different walkers where $j=1,3,5$.}
\caption{Results of an {\tt emcee} run on an AR(1) target distribution with
$\alpha = 0.9$. Each run consisted of 200,000 iterations (each iteration
being a sweep over all walkers).
The results are shown for $n=10$, $n=50$, and $n=100$ dimensional versions
of the target distribution.
}
\label{fig:fig}
\end{figure}

\begin{table}
\caption{The chosen values of the problem dimensionality $n$, along with 
the obtained mean $\hat{\mu}_n$ and standard deviation $\hat{\sigma}_n$ of the effective sample (last half of the run), where the initial conditions were generated by drawing each coordinate of each walker independently from a over-dispersed normal $N(0,10^2)$ distribution. The mean $\hat{\mu}$ and the variance $\hat{\sigma}$ were accurate
for $n=10$ and $n=50$ but not for $n=100$.}
\label{tab:gaussian_mean}
\vspace{-15pt}
\begin{center}
\begin{tabular}{r c c }
\hline
   $n$        & $\hat{\mu}_n$  &$\hat{\sigma}_n$ \\
   \hline
 10 & -0.0163909857436 & 1.05804127909 \\
 50 & 0.0104536594961 & 0.971905913468 \\
 100 & -0.498700028573 & 0.690744763472 \\               
       \hline
\end{tabular}
\end{center}
\end{table}


The output from each run consists of a three-dimensional array 
with dimensions (number of walkers, number of iterations/thinning factor,
number of dimensions).
To visualise the convergence 
properties, this array was ``flattened'' to an array which only contains the values of 
variable $x_1$.  The process of ``flattening'' reduces a two 
dimensional array which contains the values of coordinate $x_1$ of all 
walkers and at every iteration to an one dimensional array. Therefore 
the final array consists of a concatenation of $ X_1^{j=1..L}(t=0), 
X_1^{j=1..L}(t=1), X_1^{j=1..L}(t=2), \dots, X_1^{j=1..L}(t=200,000)$ in 
this specific order.

Figure~\ref{fig:1a} displays the traceplot of the 
flattened array of coordinate $x_1$ for $n=10$, where the dashed line 
displays the running average over the last $50\%$ of the elapsed time,
and the dash-dotted line displays the running standard 
deviation over the last $50\%$ of the elapsed time. The running 
average and running standard deviation were chosen because it removes 
the first $50\%$ of the ensemble, and therefore excludes more of the 
output (as potential ``burn-in'') over time.  Figure~\ref{fig:1a} 
displays a short burn-in, and it seems to come to an equilibrium quite 
quickly. Figure~\ref{fig:1d} displays empirical values of the mean of 
$x_1$ and the variance $x_1$ of all walkers as a function of time.

Figure~\ref{fig:1g} displays the joint distribution of coordinates $x_1$ and  $x_2$  of the entire ensemble for the second half of the run and it clearly shows the correlation between the two coordinates 
as expected. Figure~\ref{fig:1j} displays the trace plot of coordinate 
$x_1$ for several different walkers and it displays reasonable mixing 
and sufficient convergence. The final mean $\hat{\mu}_{n=10}$ and 
$\hat{\sigma}_{n=10}$ measured over the last $50\%$ of the chain 
displayed in Table~\ref{tab:gaussian_mean} are close to the desired values.\\

If we perform a similar analysis in higher dimensions, for example 
$n=50$, the results are close to the desired values.  Visually, the 
traceplot in Figure~\ref{fig:1b} seems to suggest that the 'fast burn-in 
period has passed at around the halfway point, and also density plot displayed in 
Figure~\ref{fig:1h} shows the correlation between the two coordinates as 
expected. Figure~\ref{fig:1e} shows that both the mean and the variance 
$x_1$ exhibit a sudden decrease from the initial standard deviation of 2.
The variance of the walkers eventually recovers, but this takes a long time.
The traceplots of the first
coordinate of for different walkers displayed in Figure~\ref{fig:1k} 
shows reasonable mixing for $n=10$ and slower mixing for $n=50$, and
$n=100$. The hope is that the large number of independent walkers
compensates for the slow movement of each walker.

For $n=100$, the results do not accurately represent the target distribution. 
To the eye, the trace plot displayed in Figure~\ref{fig:1c}
seems to show a successful MCMC run. However, the density plot of $x_1$ and $x_2$ (Figure~\ref{fig:1i}) has too small variance and correlation.
The graphs of the variance of $x_1$ for the walkers,
displayed in Figure~\ref{fig:1f}, shows that an initial sudden drop in variance
which again takes a very long time to recover. Apart from the initial fast transient, at no time
in the run of 200,000 iterations (40 million likelihood evaluations)
was the standard deviation of the
walkers' first coordinates greater than 1.
The traceplot for $n=100$ (Figure~\ref{fig:1l}) shows poor mixing, and
the estimate  $\hat{\sigma}$ is too small.\\
 
This is the main reason why the AIES should be used with caution in high
dimensions.
The output can resemble a successful run while in reality, the algorithm is still going through an initial transient phase that takes a long time.
Therefore the final sample does not represent the target distribution properly.

Roughly speaking, the burn-in process of the stretch move appears to
have two distinct stages: a fast initial transient, followed by a much slower phase. 
As we shall explain in Section~\ref{sec:taies}, the fast stage reflects convergence among the ``bulk'' of the coordinates to be consistent with the correlation structure of the AR(1) distribution.
However, the stretch moves performed during this fast stage have serious and undesirable side-effects for the ensemble of first coordinates.

\section{Convergence Diagnostics for Ensemble methods}
In theory, if a Markov chain Monte Carlo method is run for a large number of iterations, the effect of initial values will decrease to zero. Ideally the initial distribution would approach the target distribution at a certain point during the run after a relatively small number of iterations. A Markov Chain is considered converged if the probability distribution
of its state is approximately the target distribution.  In principle, the crux is to estimate the number of iterations $T$ sufficient for convergence a priori. In practise, however, it is more convenient
to try to estimate whether convergence has been achieved by
examining the output itself.
Based on the assumption that it takes $T$ iterations for the chain to converge, a chain is usually run for some number of iterations much greater than $T$ (such as $2T$) to obtain usable output.

In this section, we analyse the results from the previous section using
formal convergence diagnostics, to see whether the failure of the AIES is
detectable using these methods. Caution is required when using single-particle MCMC convergence diagnostics, since the individual walker sequences $X(t)_l$  might not be independent, or even Markovian. In general, a walker sequence and the entire ensemble do not converge at the same rate. A straightforward way to make a convergence diagnostics applicable to an ensemble method is by using a function which combines the information from all walkers into a single number, and apply the diagnostics to the obtained results. 
While the sequence of values of this summary function does not have the Markov
property, much of the reasoning behind convergence tests still applies,
at least approximately.
The obvious choice for this function would be the average or the variance of a
coordinate, taken over all the walkers.  At each iteration $t$, for each run $m$, and each parameters $n$ the average taken over all walkers is defined by:
\begin{equation}
\hat{\mu}(t)^{(i)}_m =  \frac{1}{L} \sum_{l=1}^{L}  \overline{X}(t)^{(i)}_{l,m}
\label{eq:average_over_walkers}
\end{equation} 
and the variance is defined by 
\begin{equation}
\hat{\sigma}(t)^{(i)}_m =  \frac{1}{L} \sum_{l=1}^{L}  (\overline{X}(t)^{(i)}_{l,m}  - \hat{\mu}(t)^{(i)}_m)^2 
\label{eq:variance_over_walkers}
\end{equation} 
where $i$ indicate the parameters, $t$ the iteration and $m$ indicates the run. This should enable users to apply any single particle MCMC diagnostics to the obtained results. \\
\indent  The Gelman-Rubin \cite{GelRub1992} method is a widely accepted diagnostic tool for assessment of MCMC convergence. However, it is designed to be applied to a single-particle method. Therefore the two functions mentioned before are used for the analysis. Since there might be a correlation between different parameters the method presented in this paper is based on the multivariate approach \cite{BroGel1998}. The Gelman-Rubin method is based on $M \geq 2$ independent chains, whose initial conditions were drawn from $M$ different overly-dispersed distributions. \\
   The process starts with independently simulating these $M$ chains, and discarding the first $T$ iterations.\\
After that, matrices $\mathbf{B}$ and $\mathbf{W}$ are constructed from $\mathbf{y}(t)^{(i)}_m$ which contains the results of any appropriate function applied to the walkers. The two functions chosen for $\mathbf{y}(t)^{(i)}_m$ in this paper are the averages over the walkers as defined in equation \ref{eq:average_over_walkers} and the variance over the walkers as defined in equation \ref{eq:variance_over_walkers}. 
Here $M$ indicates the number of chains, and $T$ the number of iterations.  Matrix $\mathbf{B}/T$ is the $n$-dimensional between-sequence covariance matrix estimate of the $n$ dimensional function values taken over all walkers $\mathbf{y}$:
\begin{equation}
 \mathbf{B}/T = \frac{1}{M-1} \sum^M_{j=1} (\mathbf{\overline{y}}_{j\cdot} - \mathbf{\overline{y}}_{\cdot \cdot} )(\mathbf{\overline{y}}_{j\cdot} - \mathbf{\overline{y}}_{\cdot \cdot} )'
  \label{eq:variance_between_chains}
   \end{equation}   

Matrix $\mathbf{W}$ is the within-sequence covariance matrix estimate of the $n$ dimensional average of the walkers $\mathbf{y}$:
\begin{equation}
 \mathbf{W} = \frac{1}{M(T-1)} \sum^M_{j=1} \sum^{T}_{t=1} (\mathbf{\overline{y}}_{jt} - \mathbf{\overline{y}}_{\j\cdot} )(\mathbf{\overline{y}}_{jt} - \mathbf{\overline{y}}_{j\cdot} )'
  \label{eq:variance_with_chains}
   \end{equation}   

Using the previously defined matrices one can calculate $\mathbf{\hat{V}}$ which is the estimate of the posterior variance-covariance matrix
   \begin{equation}
 \mathbf{\hat{V}} = \frac{T-1}{T}\mathbf{W}+ \left(\frac{M+1}{M} \right)\frac{\mathbf{B}}{T}
  \label{eq:posterior_variance_covariance}
   \end{equation}

The quantity of interest to establish convergence is the rotationally invariant distance measure between
$\mathbf{\hat{V}}$ and $\mathbf{W}$  \cite{BroGel1998}. This distance measure is the maximum scale reduction factor (SRF) of any linear projection of $\mathbf{y}$, and is given by
\begin{equation}
\hat{R}^n = \frac{T-1}{T} + \left(\frac{M+1}{M} \right) \lambda_1
\label{eq:MPSRF}
\end{equation}
where $\lambda_1$ is the largest eigenvalue of the positive matrix 
\begin{equation}
\mathbf{W^{-1}} \mathbf{B}/T
\label{eq:pos_def}
\end{equation}. The multivariate potential scale reduction factor $\hat{R}^n$ should approach 1 from above as $\lambda_1 \rightarrow 0$ for convergence. These computations are impossible if $\mathbf{W}$ is a singular matrix, and the results will suffer severe inaccuracies if $\mathbf{W}$ is close to being singular. 
The standard method to obtained the eigenvalues of $\mathbf{W^{-1}B}$ involves solving $\mathbf{B} = \mathbf{WX}$, however using standard software packages like \verb"eigen" in \verb"R" or \verb"numpy.linalg.eigvals" in \verb"Python" are likely to suffer from numerical instability. For efficiency and numerical stability our analysis uses a Cholensky decomposition similar to the Gelman-Rubin diagnostics in CODA \cite{CODA2006}.\\
\indent The starting distribution can still influence the final distribution after many iterations \cite{GelRub1992}.
Therefore the Gelman-Rubin method demands the starting distribution to be over-dispersed.
In an ensemble method this condition is more subtle.  One can sample each walker in such way it represents an overly-dispersed distribution, however if you do this for each of the $M$ runs used in the Gelman-Rubin method each run will still represent the same distribution, and it will be challenging for the Gelman-Rubin method to determine any lack of convergence. Therefore we propose that the starting positions of the walkers of each run are sampled from $M$ different distributions. This should enable the method to detect if each of the ensembles migrates toward the target distribution.   \\

\begin{table}[!ht]
\caption{Results of diagnostics for the Correlated Gaussian Problem, where 
the initial conditions were generated by drawing each coordinate of each walker independently from four different normal distributions --- $N(0,5^2), N(1,5^2), N(-1,5^2), N(0,10^2)$. The results of the diagnostics applied to the mean over the walkers $\hat{\mu}$ and the variance over the walkers $\hat{\sigma}$, which show  the Multivariate Ensemble PSRF  obtained from Gelman-Rubin diagnostics and Heidelberger-Welch (CODA).  }
\label{tab:gaussian_MPSRF}
\vspace{-15pt}
\begin{center}
\begin{tabular}{r| c c|  c c }
$n $  & $ \hat{R}_{\hat{\mu}}^{n}$ & $\hat{R}_{\hat{\sigma}}^{n}$  &  H-W $\hat{\mu}$  &H-W $\hat{\sigma}$   \\
\hline
 10&  1.005    &  1.009 & PASSED  & PASSED  \\
 50&  1.233    &  1.121 & PASSED & FAILED \\
100 & 2.238    &  1.688 & PASSED & FAILED
\end{tabular}
\end{center}
\end{table}

For each set of parameters $n=10$, $n=50$, and $n=100$ we performed 4 independent runs with 4 different initial conditions drawn from Gaussian distributions $N(0,5^2)$, $N(1,5^2)$,$N(-1,5^2)$ and $N(0,10^2)$. 
 The values of the multivariate potential scaled reduction factors $\hat{R}^n_{\hat{\mu}}$ and $\hat{R}^n_{\hat{\sigma}}$ are displayed in table \ref{tab:gaussian_MPSRF}.
For $n=10$ both $\hat{R}^n_{\hat{\mu}}$ and $\hat{R}^n_{\hat{\sigma}}$ are close to one which suggests convergence was achieved.
For $n=50$, $\hat{R}^n_{\hat{\mu}}$ and $\hat{R}^n_{\hat{\sigma}}$ are somewhat close to 1. However, for $n=100$ $\hat{R}^n_{\hat{\mu}}$ and $\hat{R}^n_{\hat{\sigma}}$ are much greater than one, and indicates a strong lack of convergence, in agreement with the conclusions reached
in the previous section.

As an additional convergence diagnostic, the Heidelberger-Welch-test, implemented in the CODA package\cite{CODA2006}  in \verb"R", was chosen.  The CODA-implementation performs two tests: The Heidelberger-Welch-test and the half-width test.  The Heidelberger-Welch-test is a convergence test which uses Cramer-von-Mises statistic to test the null hypothesis that the sampled values come from a stationary distribution \cite{CODA2006}. The test is initially applied to the whole chain, however if the chain fails the test a percentage at the beginning of the chain is discarded. This process is repeated until either the test is passed or 50\% percent is discarded. 
The half-width test calculates a 95\% confidence interval for the mean, using the portion of the chain which passed the stationarity test, and the calculates the half the width of this interval which is compared with the estimate of the mean. 
If the ratio between the half-width and the mean is lower than eps, the halfwidth test is passed \cite{CODA2006}. 
In this research the results of the halfwidth test are considered of little interest, because they are very subjective due to the dependence on the choice of the $\epsilon$-parameter value. 



The Heidelberger-Welch-test is applied to chains of the mean $\hat{\mu}$ and the variance $\hat{\sigma}$ as defined in equations \ref{eq:average_over_walkers} and \ref{eq:variance_over_walkers} where the first 50 \% is already discarded. Therefore we consider this test passed only if it passes this without discarding any more of the chain.  
 Only for $n=10$ did both the mean $\hat{\mu}$ and $\hat{\sigma}$ chains pass the Heidelberger-Welch test without discarding any part of the chain beyond the 100,000 iteration burn-in. These results are summarized in table \ref{tab:gaussian_MPSRF}. 

Even though it is good practice for every user of MCMC methods to use some convergence diagnostics on the obtained samples, we strongly suggest that in the case of the AIES it is not only good practice, but a necessity, since visual inspection of the results might be deceiving in high dimensions. \\


%

\section{Rosenbrock example}\label{sec:rosenbrock}
As a second test on the convergence properties of the stretch move,
we investigated an $n$-dimensional generalisation of the Rosenbrock density
\citep{Ros1960}, in the form proposed by \cite{DixMil1994}. The target
density is 
\begin{equation}
f(\mathbf{x}) = \sum_{i=1}^{n/2} 100 ( x_{2i-1}^2 - x_{2i})^2 + (x_{2i-1} - 1)^2.
\label{eq:Rosenbrock} 
\end{equation}

This is simply $n/2$ independent replications of a two-dimensional
Rosenbrock density, and should be fairly straightforward to sample.
Again, we used 200,000 iterations, but increased the number of walkers
to $L=10n$. We tested dimensionalities of $n=10$, $n=50$, and $n=100$,
so the corresponding overall numbers of likelihood evaluations were
$2 \times 10^7$, $1 \times 10^8$, and $2 \times 10^8$ respectively.

Inspection of the graphs displayed in \ref{fig:Rosenbrock} show proper mixing for $n=10$, however the mixing for $n=50$ and $n=100$ is far from desirable. Upon visual inspection the flattened traceplots for all dimensions  \ref{fig:Rosenbrock_1a}, \ref{fig:Rosenbrock_1b}, \ref{fig:Rosenbrock_1c}  shows  no indicator of convergence problems. 
 The plots of the running variance and running mean 
\ref{fig:Rosenbrock_1d}, \ref{fig:Rosenbrock_1e}
seem to indicate slow but steady convergence, however \ref{fig:Rosenbrock_1f}  doesn
t look promising. \\
A good indicator for convergence is provided from the binned scatter plots of the means of two parameters taken over the walkers  \ref{fig:Rosenbrock_1g}, \ref{fig:Rosenbrock_1h}, \ref{fig:Rosenbrock_1i}. This suggests for $n=10$ to target distribution is properly sampled, however the plots for $n=50$ and $n=100$ display  some features which might indicate convergence issues.
The individual traceplots for different parameters \ref{fig:Rosenbrock_1j}, \ref{fig:Rosenbrock_1k} and \ref{fig:Rosenbrock_1l} fail to reveal any underlying problems. \\
    The minimum scaled reduction factors  $\hat{R}^n_{\hat{\mu}}$ and $\hat{R}^n_{\hat{\sigma}}$ displayed in table \ref{tab:rosenbrock_MPSRF} suggests a lack of convergence for all models which is unexpected for $n=10$ which seem to converge properly according to the graphs displayed in \ref{fig:Rosenbrock_1a} and \ref{fig:Rosenbrock_1g}. 

While the AIES fails (at least for $n=50$ and $n=100$), other methods succeed on this problem. Simple single-particle Metropolis, with
a scale mixture of gaussians (around the current position) as the
proposal, succeeds with the equivalent computational cost
($2 \times 10^8$ likelihood evaluations), and
produces about 100 effectively independent samples, by inspection
of the empirical autocorrelation function. Of course, if the
AIES had converged, the final state of the walkers would have yielded 1000
independent samples.

\begin{table}[!ht]
\caption{The Multivariate PSRF obtained from Gelman-Rubin diagnostics for the Rosenbrock Problem, where 
the initial conditions were generated by drawing each coordinate of each walker independently from four different normal distribution $N(0,5), N(1,5), N(-1,5), N(0,10)$. The results of the diagnostics applied to the mean over the walkers $\hat{\mu}$ and the variance over the walkers $\hat{\sigma}$, which show  the Multivariate PSRF  obtained from Gelman-Rubin diagnostics and Heidelberger-Welch (CODA).}
\label{tab:rosenbrock_MPSRF}
\begin{center}
\begin{tabular}{r | c c| c  c }
$n $  & $ \hat{R}_{\hat{\mu}}^{n}$ & $\hat{R}_{\hat{\sigma}}^{n}$  & H-W $\hat{\mu}$  &H-W $\hat{\sigma}$   \\
\hline
 10&    1.74      &   2.11    &  FAILED &  PASSED \\
 50&   216        &  137      &  FAILED  &  FAILED \\
100 &   930    &   330     &  FAILED  &  FAILED \\ 
\end{tabular}
\end{center}
 \end{table}


%
%

\begin{table}
\caption{Estimates of the expected value and standard deviation
of $x_1$ in the Rosenbrock problem. The true values are approximately
1.0 and 0.7.
}
\label{tab:rosenbrock_mean}
\begin{center}
\begin{tabular}{ r c c }
 \hline 
  $n$        & $\hat{\mu}_n$  &$\hat{\sigma}_n$  \\ 
    10 & 0.912880664359 & 0.775798053317 \\
  50 & 0.219674347615 & 0.690989906816 \\
  100 & 0.397340115775 & 0.741348730082 \\   
  \hline
\end{tabular}
\end{center}
 \end{table}

\newpage
\begin{figure}
\begin{subfigure}{0.31\textwidth}
\includegraphics[width=0.94\linewidth]{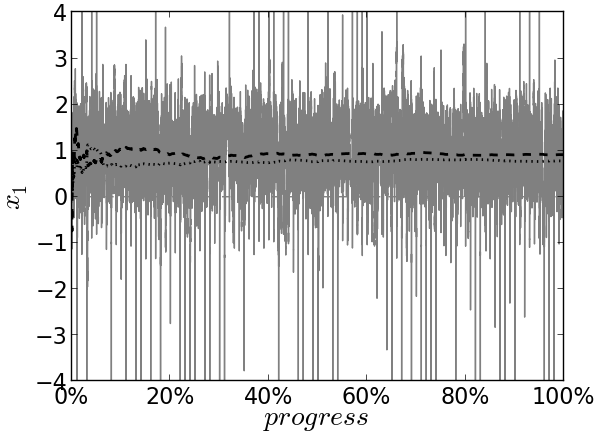}
\caption{$n=10$}
\label{fig:Rosenbrock_1a}
\end{subfigure}
\begin{subfigure}{0.31\textwidth}
\includegraphics[width=0.94\linewidth]{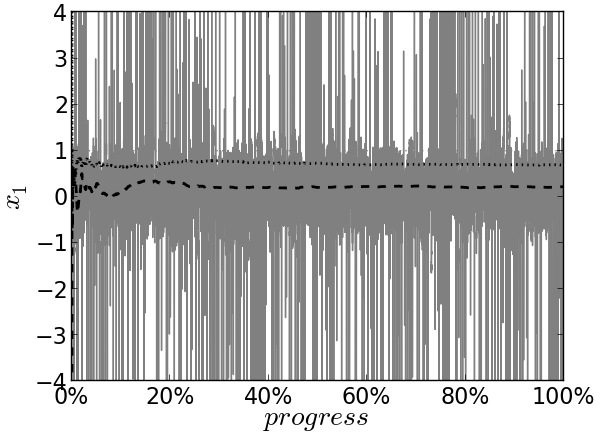}
\caption{$n=50$}
\label{fig:Rosenbrock_1b}
\end{subfigure}
\begin{subfigure}{0.31\textwidth}
\includegraphics[width=0.94\linewidth]{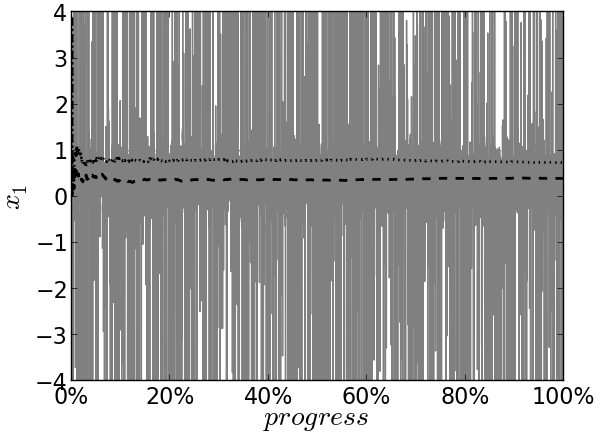}
\caption{$n=100$}
\label{fig:Rosenbrock_1c}
\end{subfigure}
\\
\vspace{-10pt}
\caption*{Graphs of the flattened trace plots of the first coordinate $x_1$ for $n=10$, $n=50$ and $n=100$, and  and $t=200.000$.}
\vspace{5pt}
\begin{subfigure}{0.31\textwidth}
\includegraphics[width=0.94\linewidth]{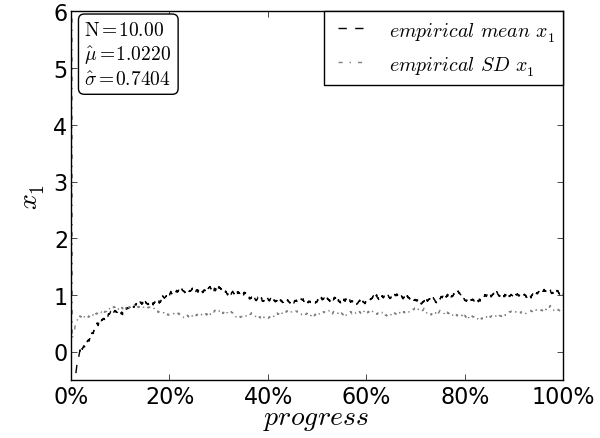}
\caption{$n=10$}
\label{fig:Rosenbrock_1d}
\end{subfigure}  
\vspace{3pt}
\begin{subfigure}{0.31\textwidth}
\includegraphics[width=0.94\linewidth]{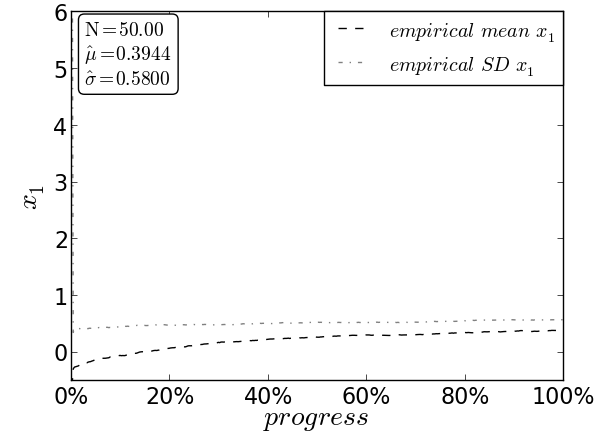}
\caption{$n=50$}
\label{fig:Rosenbrock_1e}
\end{subfigure}  
\vspace{3pt}
\begin{subfigure}{0.31\textwidth}
\includegraphics[width=0.94\linewidth]{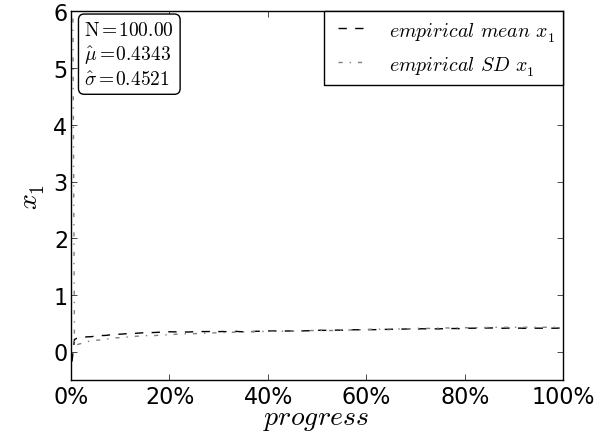}
\caption{$n=100$}
\label{fig:Rosenbrock_1f}
\end{subfigure} 
 \\
\vspace{-10pt}
\caption*{Mean values (grey) and variance (black) of $x_1$ over all walkers as a function of iterations $t$.}
\vspace{5pt}
\hspace{-5pt}

\begin{subfigure}{0.31\textwidth}
\includegraphics[width=0.99\linewidth]{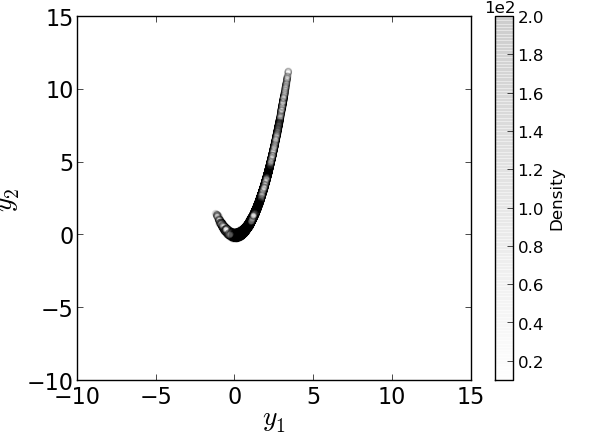}
\caption{$n=10$}
\label{fig:Rosenbrock_1g}
\end{subfigure} 
\begin{subfigure}{0.31\textwidth}
\vspace{0pt}
\includegraphics[width=0.99\linewidth]{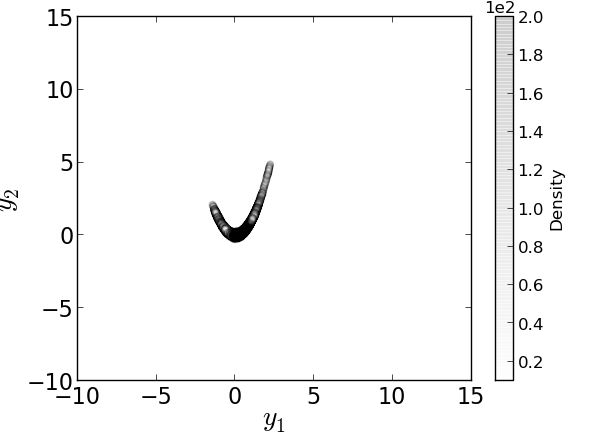}
\caption{$n=50$}
\label{fig:Rosenbrock_1h}
\end{subfigure} 
\begin{subfigure}{0.31\textwidth}
\vspace{0pt}
\includegraphics[width=0.99\linewidth]{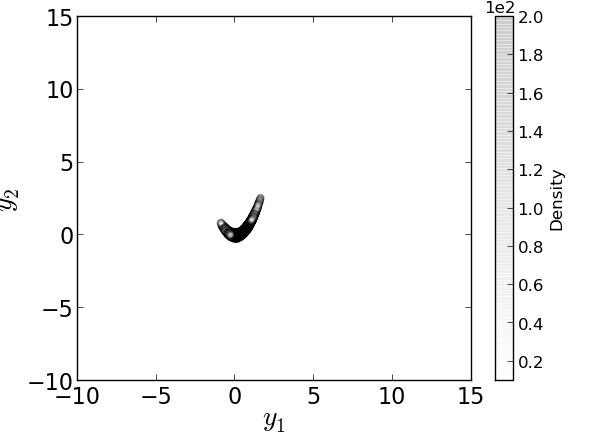}
\caption{$n=100$}
\label{fig:Rosenbrock_1i}
\end{subfigure}  
\\
\vspace{-10pt}
\caption*{Density plot of coordinates $x_1$ and $x_2$ from the second
half of the run.}
\vspace{5pt}
\hspace{-5pt}
\begin{subfigure}{0.31\textwidth}
\includegraphics[width=0.99\linewidth]{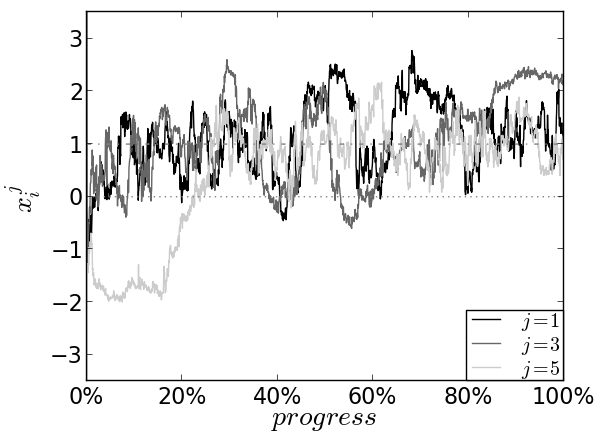}
\caption{$n=10$} 
\label{fig:Rosenbrock_1j}
\end{subfigure}
\begin{subfigure}{0.31\textwidth}
\includegraphics[width=0.99\linewidth]{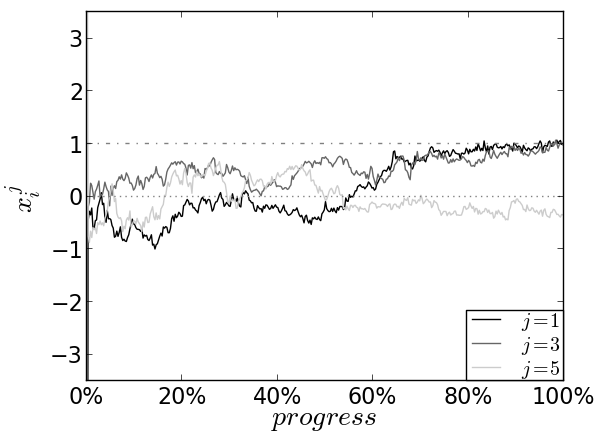}
\caption{$n=50$} 
\label{fig:Rosenbrock_1k}
\end{subfigure}  
\begin{subfigure}{0.31\textwidth}
\includegraphics[width=0.99\linewidth]{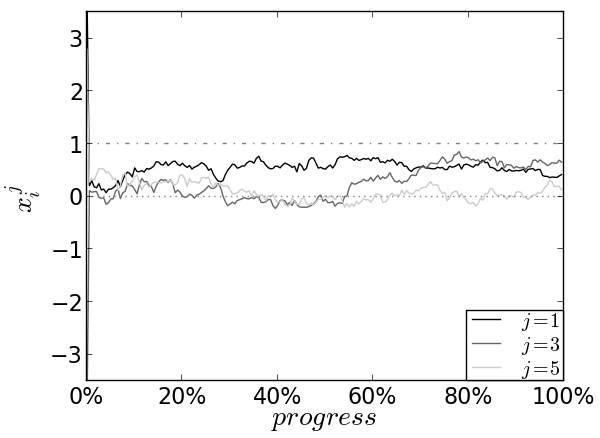}
\caption{$n=100$} 
\label{fig:Rosenbrock_1l}
\end{subfigure} 
\vspace{-10pt}
\caption*{Trace plot of the coordinate $x_1$ for different walkers where $j=1,3,5$.}
\caption{Results of an MCMC run of a Rosenbrock distribution using
200,000 iterations, for dimensionalities of $n=10, 50, 100$.
The number of walkers was $10n$ in each case. }
\label{fig:Rosenbrock}
\end{figure}

\section{Theoretical causes of the behaviour of the AIES for sampling a high dimensional Gaussian}\label{sec:taies}

The results in Section~\ref{sec:gaussian} suggest studying the limiting behaviour of the AIES in an appropriate limit as $n\to\infty$.
This limit is somewhat complicated, not least because it requires the number of walkers to be large since $L\geq n+1$.
We begin with a description of the AIES in the limit $L\to\infty$ for fixed $n$ and $\pi$.
Then we examine a single AIES move in the limit $n\to\infty$ under a simplifying assumption.
We then give a non-rigorous heuristic for the limit $n\to\infty$ that explains the behaviour described in Section~\ref{sec:gaussian}.

\subsection{The AIES with many walkers}\label{ss:AIESmanyWalkers}

To study the limit $L\to\infty$, it is convenient to use the continuous-time variant where the main and complementary walkers are selected uniformly among all walkers.
Then the $L$ walkers play symmetric roles and it is natural to collect them into the \emph{empirical measure}
\begin{equation}
\mu^{(L)}(t)=\frac{1}{L}\sum_{j=1}^L \delta_{\mathbf{X}^{(j)}(t)},
\end{equation}
where $\delta_\mathbf{x}$ denotes the measure placing unit mass at $\mathbf{X}\in\mathbb{R}^n$.
In words, the measure $\mu^{(L)}(t)$ encodes the distribution of a uniformly chosen walker at time $t$.
Because of the assumption that both walkers $X$ and $Y$ are selected uniformly, it follows that $\mu^{(L)}(t)$ is itself a Markov chain.

\begin{prop}
\label{prop:AIESmany}
Choose the initial walkers independently according to the distribution $\mu_0$.  
Then, in the limit $L\to\infty$, the empirical measure process $\mu^{(L)}(t)$ converges in distribution to a \emph{deterministic} path $\mu_t$ with initial value $\mu_0$ and 
\begin{equation}\label{AIESmanyODE}
\frac{d}{dt}\int f(\mathbf{x}) d\mu_t(\mathbf{x}) = \iint_{\R^n\times\R^n} \E\Big(  f\big(Z\mathbf{x}+(1-Z)\mathbf{y}\big)-f(\mathbf{x}) \Big) p(\mathbf{x},\mathbf{y},\mathbf{z}) d\mu_t(\mathbf{x}) d\mu_t(\mathbf{y})
.
\end{equation}
\end{prop}

To interpret this result, suppose $\mathbf{X}$ and $\mathbf{Y}$ are independent samples distributed according to the current empirical measure $\mu_t$.
Choose $Z$ according to the density in \eqref{eq:Zdensity} and define $\widetilde{\mathbf{X}}=\mathbf{X}Z+\mathbf{Y}(1-Z)$, as in \eqref{eq:Proposal}. 
Set $\mathbf{X}'$ to be $\widetilde{\mathbf{X}}$ with probability $p(\mathbf{X},\mathbf{Y},Z)$ and $\mathbf{X}$ otherwise, and let $\mu'_t$ denote the measure encoding the distribution of $\mathbf{X}'$, averaged over all the possibilities for $\mathbf{X},\mathbf{Y}$ and $Z$.
Then Proposition~\ref{prop:AIESmany} says that the measure $\mu_t$ evolves by travelling in the direction of the line (in the space of measures) joining $\mu_t$ to $\mu'_t$.

The intuition behind Proposition~\ref{prop:AIESmany} is that the average effect of each move is to take $\mu_t$ in the direction toward $\mu'_t$.
Each move changes a fraction $1/L$ of the measure $\mu^{(L)}(t)$, but this is offset by the fact that moves occur at rate $L$.
The fact that the limiting dynamics are deterministic is established by examining products of empirical means.

The authors did not find any explicit solutions to the system of equations~\eqref{AIESmanyODE}.
(In the simplest case $n=1$, $\pi(\mathbf{x})\propto \exp\left(-\mathbf{x}^2/2 \right)$, i.e., a one-dimensional standard normal density, any normal initial data that is not standard becomes non-normally distributed at positive times and does not appear to follow any simple trajectory.)
However, it is possible to consider \eqref{AIESmanyODE} for other choices for the function $p$.
The simplest possible choice is to take $p=1$, i.e., to accept proposals unconditionally.
Even in this case, we still cannot solve \eqref{AIESmanyODE}, but we can make the following observation.

\begin{prop}\label{prop:AlwaysAccept}
Consider the system of equations \eqref{AIESmanyODE} where the function $p(\mathbf{x},\mathbf{y},z)$ is replaced by the constant 1.
If the $i^\text{th}$ coordinate has finite second moment under the initial measure, $\int x_i^2 d\mu_0(\mathbf{x})<\infty$, then its mean $\int x_i d\mu_t(\mathbf{x})$ is constant and its variance $\Var_{\mu_t}(x_i)=\int x_i^2 d\mu_t(\mathbf{x})-(\int x_i d\mu_t(\mathbf{x}))^2$ evolves according to
\begin{equation}
\frac{d}{dt} \Var_{\mu_t}(x_i) = \E\big( Z\left(Z-1\right) \big) \Var_{\mu_t}(x_i).
\end{equation}
\end{prop}

Thus the variance will either grow or decay exponentially, depending on whether $\E\big(Z\left(Z-1\right)\big)$ is positive or negative.

\subsection{The AIES for a high-dimensional standard Gaussian}\label{ss:AIESstandardGaussian}

We next consider a single stretch move for the target density
\begin{equation}
\pi'(\mathbf{x})=c\cdot \exp\left(-\frac{1}{2}\sum_{i=1}^n x_i^2 \right)
\end{equation}
corresponding to an $n$-dimensional standard normal distribution.
The acceptance probability from \eqref{eq:AcceptanceProb} becomes $p(\mathbf{x},\mathbf{y},z)=\min\set{1,\exp\left(h\left(\mathbf{x},\mathbf{y},z\right)\right)}$ where
\begin{equation}
h(\mathbf{x},\mathbf{y},z) = (n-1)\log z - \frac{1}{2}\sum_{i=1}^n\big( z x_i+ \left(1-z\right)y_i\big)^2 + \frac{1}{2}\sum_{i=1}^n x_i^2.
\label{eq:AcceptNormal}
\end{equation}
To analyze $h(\mathbf{x},\mathbf{y},z)$, we make the following assumption about the randomly chosen walkers $\mathbf{X}$ and $\mathbf{Y}$ for the move under consideration.

\begin{assumption}\label{IIDassumption}
The coordinates of $X_1,\dotsc,X_n$ and $Y_1,\dotsc,Y_n$ are mutually independent and identically distributed (i.i.d.) with common mean $\mu$ and common variance $\sigma^2$.
\end{assumption}

\begin{prop}\label{prop:phScaling}
Subject to Assumption~\ref{IIDassumption}, $\frac{1}{n}h(\mathbf{X},\mathbf{Y},z)\to f_\sigma(z)$ where
\begin{equation}
f_\sigma(z)=\log z-\sigma^2 z(z-1).
\end{equation}
The acceptance probability $p(\mathbf{X},\mathbf{Y},z)$ converges to 1 if $f_\sigma(z)>0$ and to 0 if $f_\sigma(z)<0$.
\end{prop}

\begin{proof}
This is an application of the Law of Large Numbers:
\begin{align}
\frac{h(\mathbf{x},\mathbf{y},z)}{n}
&=
-\frac{\log z}{n}+\frac{1}{n}\sum_{i=1}^n \left(\log z- \frac{1}{2}(z^2 -1)x_i^2 - z(1-z)x_i y_i-\frac{1}{2}(1-z)^2 y_i^2 \right)
\notag\\&
\to 0+\E\left( \log z- \frac{1}{2}(z^2-1)x_i^2 - z(1-z)x_i y_i-\frac{1}{2}(1-z)^2 y_i^2 \right)
\notag\\&
= \log z -\frac{1}{2}(z^2-1)(\sigma^2+\mu^2)-z(1-z)\mu\cdot\mu-\frac{1}{2}(1-z)^2(\sigma^2+\mu^2)
\notag\\&
= \log z - \sigma^2 \frac{z^2-1+(1-z)^2}{2} - \mu^2 \frac{z^2-1+2z(1-z)+(1-z)^2}{2}
\notag\\&
= f_\sigma(z)
.
\end{align}
The convergence of $p(\mathbf{X},\mathbf{Y},z)=\min\left(1,\exp(h(\mathbf{X},\mathbf{Y},z))\right)$ follows because either $h(\mathbf{X},\mathbf{Y},z)\to\infty$ or $h(X,Y,z)\to -\infty$ depending on the sign of $f_\sigma(z)$.
\end{proof}

The behaviour of $f_\sigma(z)$ for three values of $\sigma$ is shown in Figure~\ref{fig:fsigmaz}.
When $\sigma>1$, $f_\sigma(z)$ is positive for $z$ slightly smaller than $1$.
When $\sigma<1$, $f_\sigma(z)$ is positive for $z$ slightly larger than $1$.
In the critical case $\sigma=1$, $f_\sigma(z)$ is always negative and a Taylor expansion gives
\begin{equation}
f_1(z) \approx -\tfrac{3}{2}(z-1)^2 \qquad\text{for $z$ close to }1.
\end{equation}

The interpretation of Assumption~\ref{IIDassumption} and Proposition~\ref{prop:phScaling} is as follows.
Freeze the AIES algorithm at time $t$.
The walkers $\mathbf{X}$ and $\mathbf{Y}$ to be used in the next move will be drawn independently from the current population of walkers.
Assume without proof that the coordinates $X_1,\dotsc,X_n$ and $Y_1,\dotsc,Y_n$ are i.i.d.\ (or at least are sufficiently close to i.i.d.\ for the conclusion of Proposition~\ref{prop:phScaling} to hold).
\begin{wrapfigure}{r}{0.5\textwidth}
\vspace{-13pt}
\includegraphics[width=0.95\linewidth]{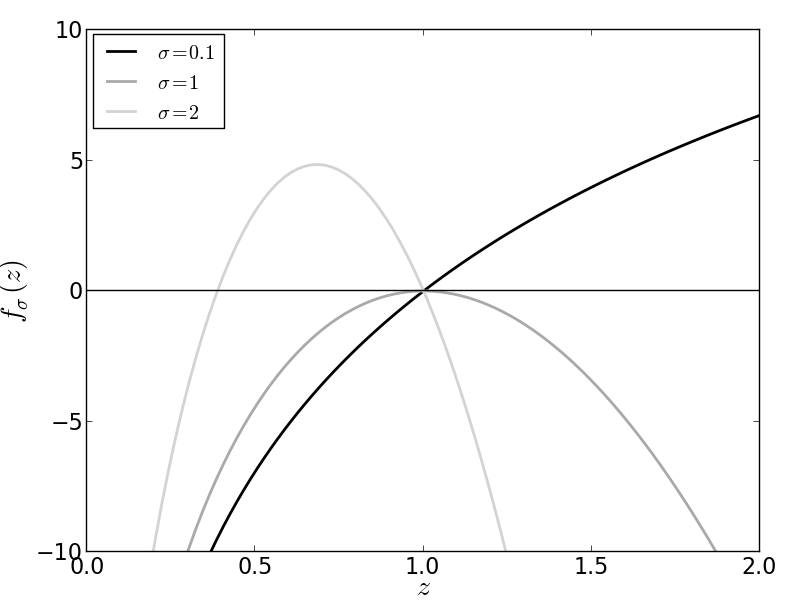}
\vspace{-10pt}
\caption{The graph of the function  $F(z) = n(\log z -\sigma^2(t) z(z-1))$ as a function of $z$ for 
$\sigma(t)=0.1$ (dashed-dotted line), $\sigma(t)=1$ (dotted line), and $\sigma(t)=2$ (dashed line).}
  \label{fig:fsigmaz}
\vspace{-5pt}
\end{wrapfigure}
Then the acceptance probability at the next move is effectively independent of the actual walkers $\mathbf{X}$ and $\mathbf{Y}$.
Furthermore the dependence on $z$ is determined only by the variance $\sigma^2$.
To illustrate this effect, we ran three different MCMC-runs on a uncorrelated Gaussian, where each run consist $1000$ repetitions of two consecutive steps.  The first step  is a (re-)initialization step where the code draws initial conditions, aand the second step consist of one regular AIES iteration. 
Each of these $1000$ AIES iteration was started with a re-initialized initial condition obtained in the first step, where each of the $n$ components of every of $L$ walkers are drawn from a Gaussian distribution $N(0,\sigma^2)$ with $\sigma^2 \in \{0.1, 1 , 2\}$.  We obtained the accepted $z$-values, and created histograms displayed in Figure~\ref{fig:HistAcceptedZRepeated}. 
  For $\sigma\neq 1$, a sharp cutoff is observed across the value $Z=1$, becoming more pronounced as $n$ increases.
Values of $z$ on the predicted side of $1$ are accepted almost unconditionally.
For $\sigma=1$, the approximation $p(\mathbf{X},\mathbf{Y},z)\approx \exp\left( -cn(z-1)^2 \right)$ suggests that the accepted values of $Z$ will be roughly normally distributed with a spread that decreases with $n$, and this prediction is reflected in the histograms.

\begin{figure}
\begin{subfigure}{0.32\textwidth}
\includegraphics[width=0.95\linewidth]{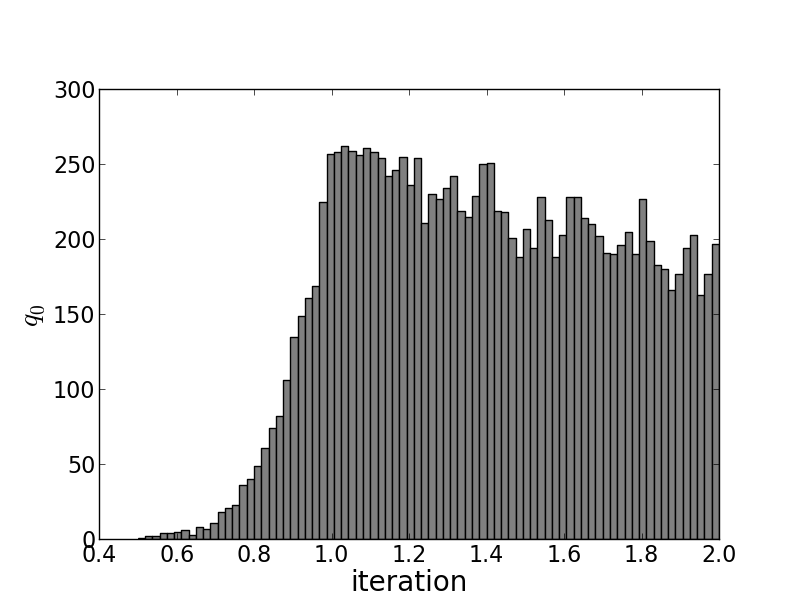}
\caption{$\sigma(0)=0.1$ for $n=10$}
\label{fig:Histzsigma0.1n10}
\end{subfigure}
\begin{subfigure}{0.32\textwidth}
\includegraphics[width=0.95\linewidth]{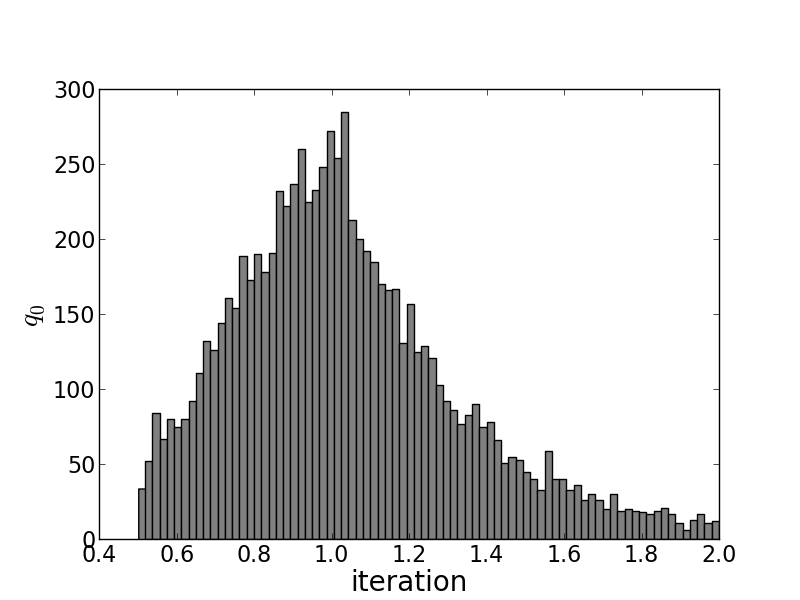}
\caption{$\sigma(0)=1.0$ for $n=10$.}
\label{fig:Histzsigma1.0n10}
\end{subfigure}  
\begin{subfigure}{0.32\textwidth}
\includegraphics[width=0.95\linewidth]{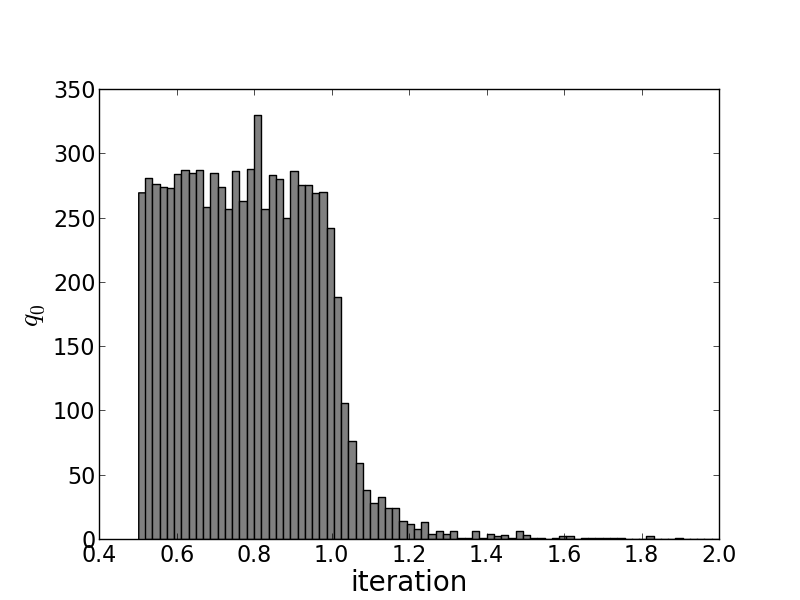}
\caption{$\sigma(0)=2.0$ for $n=10$.}
\label{fig:3c}
\end{subfigure}  \\

\begin{subfigure}{0.32\textwidth}
\includegraphics[width=0.95\linewidth]{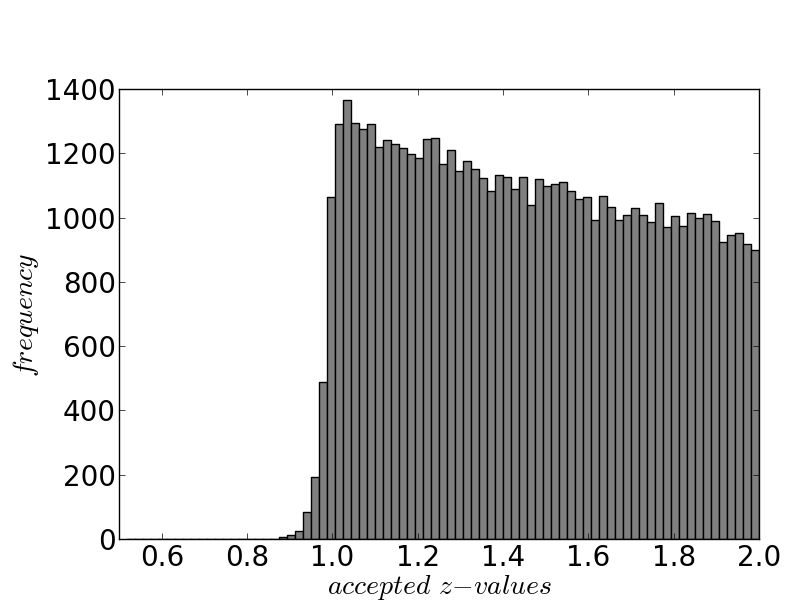}
\caption{$\sigma(0)=0.1$ for $n=50$.}
\label{fig:3d}
\end{subfigure}
\begin{subfigure}{0.32\textwidth}
\includegraphics[width=0.95\linewidth]{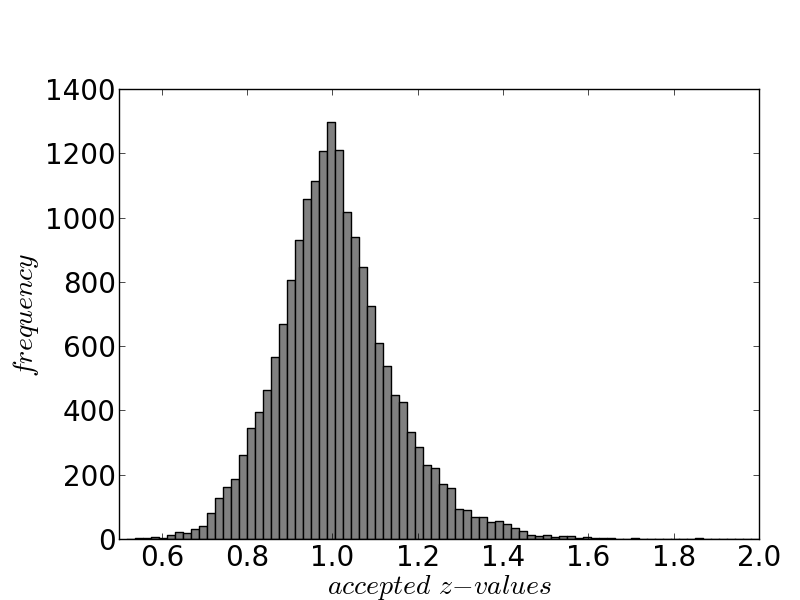}
\caption{$\sigma(0)=1.0$ for $n=50$.}
\label{fig:3e}
\end{subfigure}  
\begin{subfigure}{0.32\textwidth}
\includegraphics[width=0.95\linewidth]{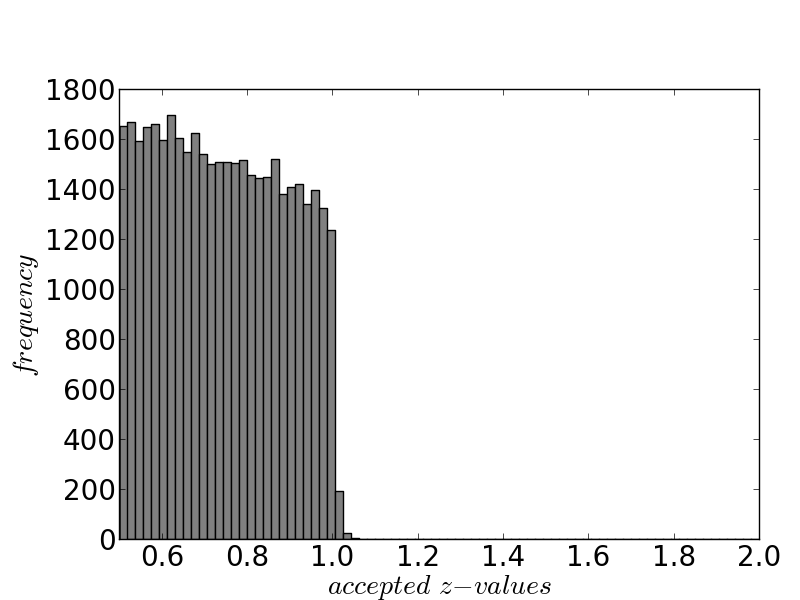}
\caption{$\sigma(0)=2.0$ for $n=50$.}
\label{fig:3f}
\end{subfigure}  \\
\begin{subfigure}{0.32\textwidth}
\includegraphics[width=0.95\linewidth]{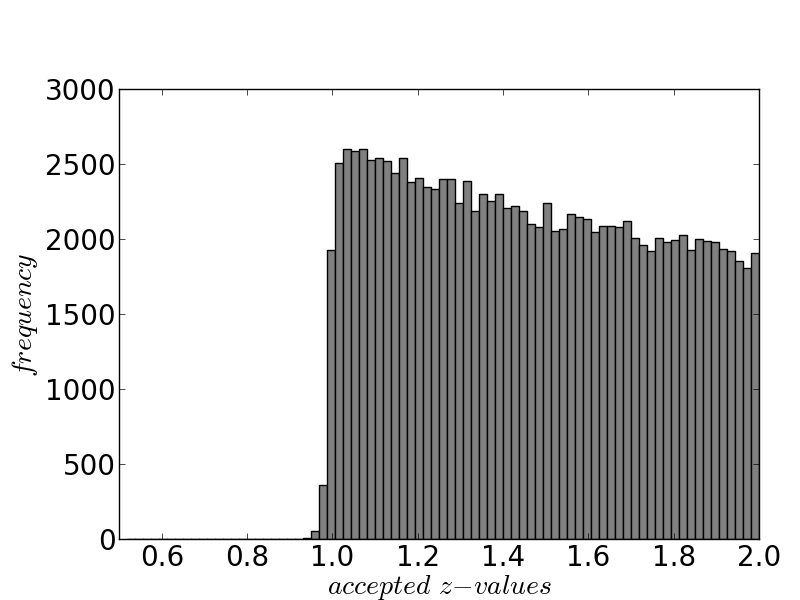}
\caption{$\sigma(0)=0.1$ for $n=100$.}
\label{fig:3g}
\end{subfigure}
\begin{subfigure}{0.32\textwidth}
\includegraphics[width=0.95\linewidth]{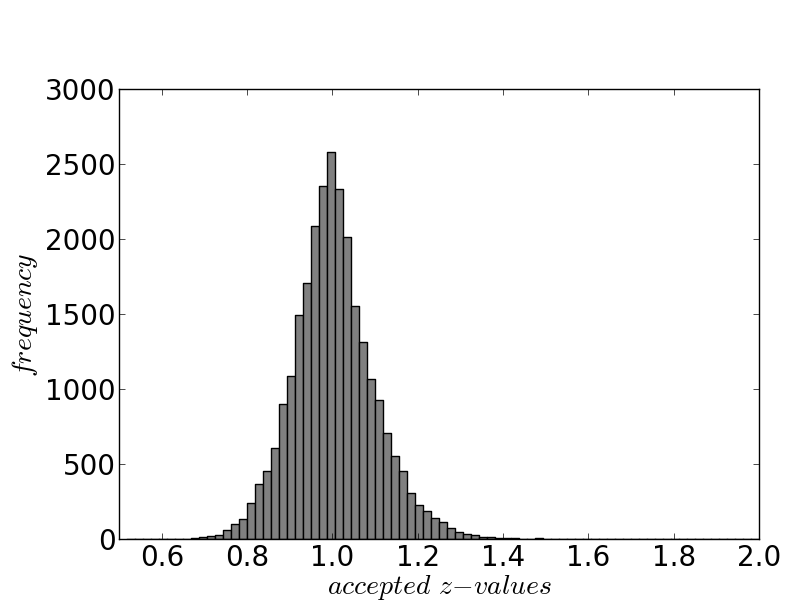}
\caption{$\sigma(0)=1.0$ for $n=100$.}
\label{fig:3h}
\end{subfigure}  
\begin{subfigure}{0.32\textwidth}
\includegraphics[width=0.95\linewidth]{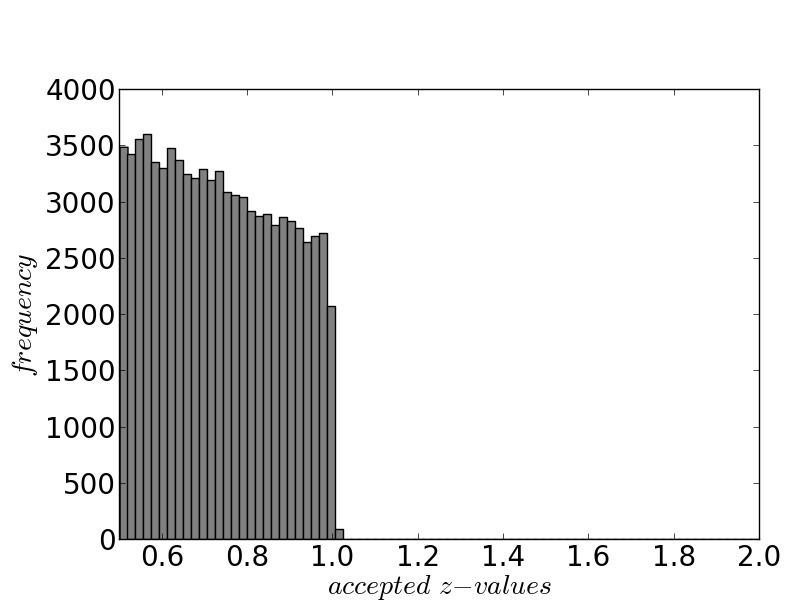}
\caption{$\sigma(0)=2.0$ for $n=100$.}
\label{fig:3i}
\end{subfigure}  
\caption{Histogram of the accepted $z$ for $n=10, 50, 100$ and $\sigma(0)=0.1,1,2$ for an MCMC run of an uncorrelated Gaussian with $t=1000$.}
\label{fig:HistAcceptedZRepeated}
\vspace{-10pt}
\end{figure}

\subsection{The AIES for a high-dimensional correlated Gaussian}
We now turn to the correlated AR(1) model from Section~\ref{sec:gaussian}.
Based on the analysis of Sections~\ref{ss:AIESmanyWalkers} and \ref{ss:AIESstandardGaussian}, we present a heuristic that explains the behaviour observed in Section~\ref{sec:gaussian}.

Because of Proposition~\ref{prop:AffineInvariance}, we can apply an affine transformation to this correlated Gaussian distribution into an uncorrelated one.
Recalling \eqref{eq:AR1distr}, the relevant transformation $\psi(x)$ will be
\begin{equation}
\begin{aligned}
 q_1 &= \psi_1(x)=x_1, & x_1 &=q_1, \\
 q_i &= \psi_i(x)=\frac{x_i-\alpha x_{i-1}}{\beta}, & x_i &=\alpha x_{i-1}+\beta q_i, && i\geq 2. 
\label{eq:xq}
\end{aligned}
\end{equation}
A random variable $\mathbf{X}$ has the AR(1) distribution if and only if the corresponding $\mathbf{Q}$ has the $n$-dimensional standard normal distribution.
This problem therefore falls in the setup of Section~\ref{ss:AIESstandardGaussian} with density $\pi'$ when expressed in terms of the $\mathbf{q}$-coordinate system.

Note that the coordinate $q_1$ plays a distinguished role, directly measuring the quantity of interest from the original system.
(Different transformations can be used to emphasise other quantities from the original system.)
The coordinates $q_2,\dotsc,q_n$ can be thought of as encoding how closely the coordinates $x_i$ conform to the correlation structure of the AR(1) distribution.

We will analyse the AIES algorithm from the perspective of the $\mathbf{q}$-coordinate system.
To begin, we must specify the initial walker coordinates.
In practice, we are unlikely to have any knowledge of the $\mathbf{q}$-coordinate system.
Therefore the most obvious choice is to generate i.i.d.\ initial coordinates $X_i^{(j)}(0)$ in the $\mathbf{x}$-coordinate system.
By \eqref{eq:xq}, the initial $\mathbf{q}$-coordinates have variances
\begin{equation}
\begin{aligned}
\Var\big(Q_1^{(j)}(0)\big)&=\Var\Big(\psi_1\big(\mathbf{X}^{(j)}(0)\big)\Big) = \Var\left(X_1^{(j)}\left(0\right)\right), \\
\Var\big(Q_i^{(j)}(0)\big)&=\Var\left(\psi_i\big(\mathbf{X}^{(j)}(0)\big)\right) = \frac{1+\alpha^2}{\beta^2} \Var\left(X_i^{(j)}\left(0\right)\right), & i\geq 2.
\end{aligned}
\label{eq:InitialVariances}
\end{equation}
Note that the variance for $i\geq 2$ is larger by a factor $\frac{1+\alpha^2}{1-\alpha^2}$ compared to $i=1$.
This factor becomes large in the highly correlated case where $\alpha$ is close to 1.

To proceed with our heuristic analysis, we introduce without proof the following assumption.

\begin{assumption}\label{AlwaysIIDassumption}
At every time $t$ in the AIES algorithm, the $\mathbf{q}$-coordinates $(Q^{(J)}_i(t),i=1,\dotsc,n)$ (possibly excluding the first coordinate) of a randomly chosen walker are i.i.d.\ with common mean $\mu(t)$ and common variance $\sigma(t)^2$ -- or at least are sufficiently close to i.i.d.\ that the conclusions of Proposition~\ref{prop:phScaling} apply.
\end{assumption}

In reality, the assumption of independence will not hold even at time $t=0$; the initial $\mathbf{q}$-coordinates are only weakly correlated in the sense that $\Cov(Q_i^{(j)}(0),Q_{i'}^{(j)}(0))=0$ if $\abs{i-i'}\geq 2$.
Even if the initial $q$-coordinates were chosen in an i.i.d.\ way, there would be no reason for the AIES dynamics to preserve this property at later times.
However, Proposition~\ref{prop:phScaling} depends only on a Law of Large Numbers effect, and at the level of a heuristic it is reasonable to expect this effect to be robust.

Subject to Assumption~\ref{AlwaysIIDassumption}, Proposition~\ref{prop:phScaling} asserts that the acceptance probability $p(\mathbf{X},\mathbf{Y},Z)$ is essentially independent of the actual walker positions $\mathbf{X}$ and $\mathbf{Y}$.
In particular, it is essentially independent of the actual first $q$-coordinates $\psi_1(\mathbf{X})$ and $\psi_1(\mathbf{Y})$.
From the perspective of the first coordinates only, the AIES dynamics are 
approximated\footnote{A careful justification of this approximation involves more than the fact that $p(\mathbf{X},\mathbf{Y},z)$ is largely independent of $\mathbf{X}$ and $\mathbf{Y}$.
                      Specifically, the justification is that the probabilities $p(\mathbf{X},\mathbf{Y},Z)$ and $p'$ are unlikely to have a large difference.
											This holds because according to Proposition~\ref{prop:phScaling} the quantity $h(\mathbf{X},\mathbf{Y},Z)$ is likely to be either large and positive (in which case it is unlikely to become negative after removing the $i=1$ term, so that both $p$ and $p'$ will be 1) or large and negative (in which case removing the $i=1$ term may make a difference, but both $p$ and $p'$ will still be small).
											This reasoning will break down when $h(\mathbf{X},\mathbf{Y},Z)$ has the possibility to be small, which will happen when $Z$ and $\sigma(t)$ are close to 1.} 
by the following:
\begin{enumerate}
\item
Select walkers $\mathbf{X}$ and $\mathbf{Y}$, stretching variable $Z$, and proposal $\widetilde{\mathbf{X}}=Z\mathbf{X}+(1-Z)\mathbf{Y}$ as usual.
Write $\mathbf{Q}=\psi(\mathbf{X})$, $\mathbf{U}=\psi(\mathbf{Y})$, and $\widetilde{\mathbf{Q}}=\psi(\widetilde{\mathbf{X}})$.

\item\label{item:AcceptProbModified}
Define the modified acceptance probability 
\begin{equation}
p' = \min\Bigg(1, Z^{n-1}\exp\bigg( -\frac{1}{2}\sum_{i=2}^n \Big( \big(ZQ_i+(1-Z)U_i\big)^2 - Q_i^2 \Big) \bigg) \Bigg)
\end{equation}
solely in terms of $\mathbf{q}$-coordinates 2 to $n$.
In coordinates 2 to $n$, update from $\mathbf{Q}$ to $\widetilde{\mathbf{Q}}$ with probability $p'$.

\item
For each move accepted in step~\ref*{item:AcceptProbModified}, apply the same move to the first coordinate with probability 1.
\end{enumerate}

From the perspective of the first coordinates, the accepted stretching variables $Z$ follow a modified distribution $\widetilde{Z}(t)$ that may vary over time as the other $q$-coordinates equilibrate.
Stretching variables $\widetilde{Z}(t)$ also arrive at a reduced average frequency $r(t)=\E(p')$.
However, when they arrive, they are always accepted, as in Proposition~\ref{prop:AlwaysAccept}.
We therefore make the following predictions:

\begin{prediction}\label{pred:VarianceFromsigmat}
Write 
\begin{equation}
\Var_t(X_1) = \frac{1}{L} \sum_{j=1}^L X_1^{(j)}(t)^2 - \biggl( \frac{1}{L}\sum_{j=1}^L X_1^{(j)}(t) \biggr)^2
\end{equation}
for the empirical variance of $X_1$ as determined by the walkers at time $t$.
Then:
\begin{itemize}
\item 
When $\sigma(t)\gg 1$, $\Var_t(X_1)$ will decrease rapidly, regardless of whether it is too large or too small compared to the true value 1.
\item
When $\sigma(t)\ll 1$, $\Var_t(X_1)$ will increase rapidly, regardless of whether it is too large or too small compared to the true value 1.
\item
When $\sigma(t)$ is close to 1, $\Var_t(X_1)$ will not converge quickly to the true value 1.
\item
Quantitatively, 
\begin{equation}
\frac{d}{dt}\Var_t(X_1) \approx r(t) \E\left( \widetilde{Z}(t)(\widetilde{Z}(t)-1) \right) \Var_t(X_1).
\label{eq:19}
\end{equation}
\end{itemize}
\end{prediction}

\begin{figure}
\begin{subfigure}{0.31\textwidth}
\includegraphics[width=0.99\linewidth]{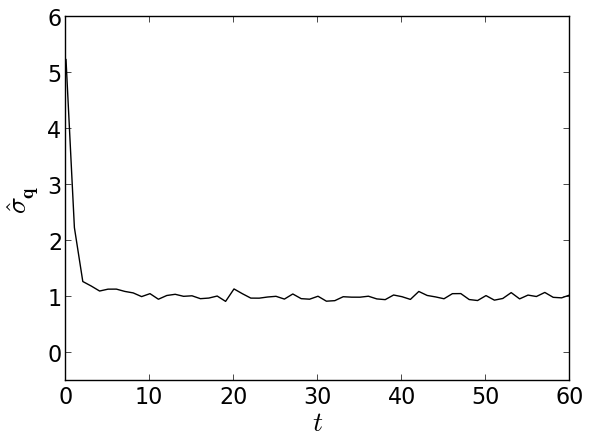}
\caption{$n=10$} 
\label{fig:4a}
\end{subfigure}
\begin{subfigure}{0.31\textwidth}
\includegraphics[width=0.99\linewidth]{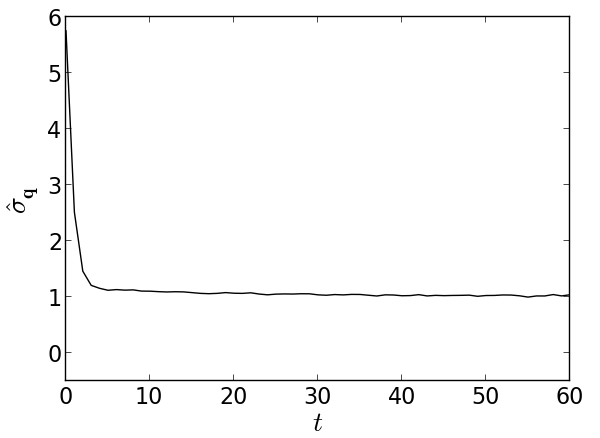}
\caption{$n=50$} 
\label{fig:4b}
\end{subfigure}  
\begin{subfigure}{0.31\textwidth}
\includegraphics[width=0.99\linewidth]{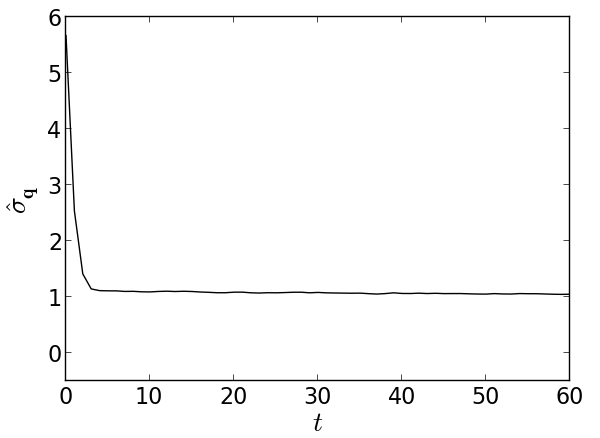}
\caption{$n=100$} 
\label{fig:4c}
\end{subfigure}  
\caption{The variance of $q$ over all walkers and coordinates for $\alpha=0.9$, $t=1000$}
\label{fig:fig4}
\end{figure}

To test Prediction~\ref{pred:VarianceFromsigmat}, we graphed the average empirical standard deviation of all $\mathbf{q}$-coordinates at times close to the fast burn-in phase: see Figure~\ref{fig:fig4}.
The initial value is greater than 1, as predicted by \eqref{eq:InitialVariances}, and has equilibrated around 1 after about 10 iterations.
In Figure~\ref{fig:fig5}, the empirical variance $\Var_t(X_1)$ of all first coordinates is graphed over the same time interval.
Especially for $n=50$ and $n=100$, the variance decreases rapidly at first, then levels off around the same time, 10 iterations.
This confirms the qualitative parts of Prediction~\ref{pred:VarianceFromsigmat}.

We also tested the quantitative prediction in equation~\eqref{eq:19}.
At five selected times $t=2,8,14, 30, 50$, we overlaid the predicted tangent line from Prediction~\ref{pred:VarianceFromsigmat} -- i.e., the line with slope given by the right-hand side of equation~\eqref{eq:19} -- onto the graph of $\Var_t(X_1)$.
The quantities $r(t)$ and $\E\left( \widetilde{Z}(t)(\widetilde{Z}(t)-1) \right)$ in equation~\eqref{eq:19} depend on the hypothetical distribution of all $Z$ values that would be accepted given the actual walkers at time $t$, and were therefore approximated by observing the accepted $Z$ values from 100 auxiliary {\tt emcee} iterations, each initialised with the actual walker positions at time $t$.

The results are shown in Figure~\ref{fig:fig5}.
For $n=10$, the predicted lines do not very closely track the underlying curve, but for $n=50$ and $n=100$ there is good agreement with Prediction~\ref{pred:VarianceFromsigmat}.

\begin{figure}[ht]
\begin{subfigure}{0.32\textwidth}
\includegraphics[width=0.99\linewidth]{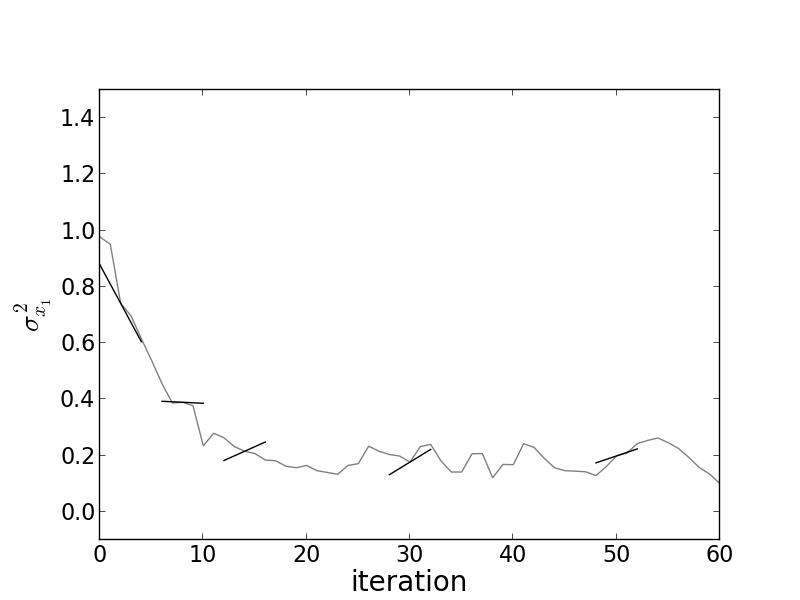}
\caption{$n=10$} 
\label{fig:5a}
\end{subfigure}
\begin{subfigure}{0.32\textwidth}
\includegraphics[width=0.99\linewidth]{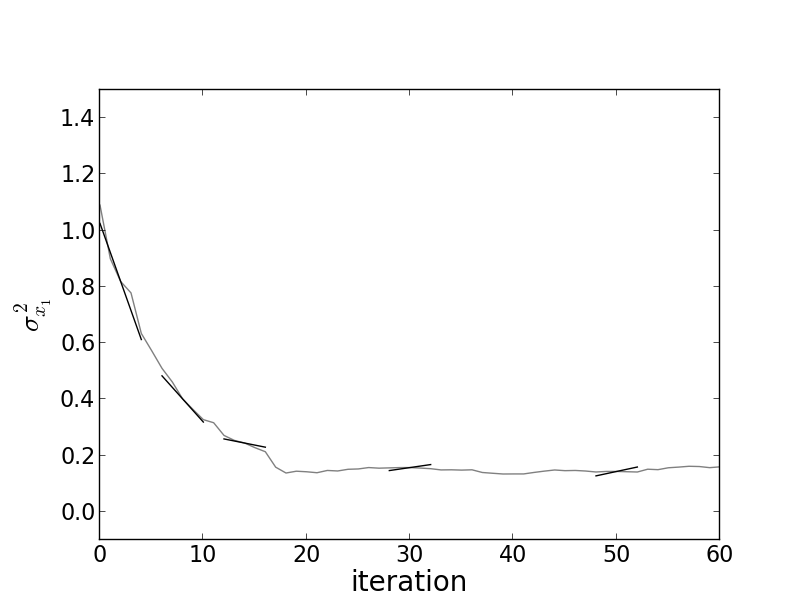}
\caption{$n=50$} 
\label{fig:5b}
\end{subfigure}  
\begin{subfigure}{0.32\textwidth}
\includegraphics[width=0.99\linewidth]{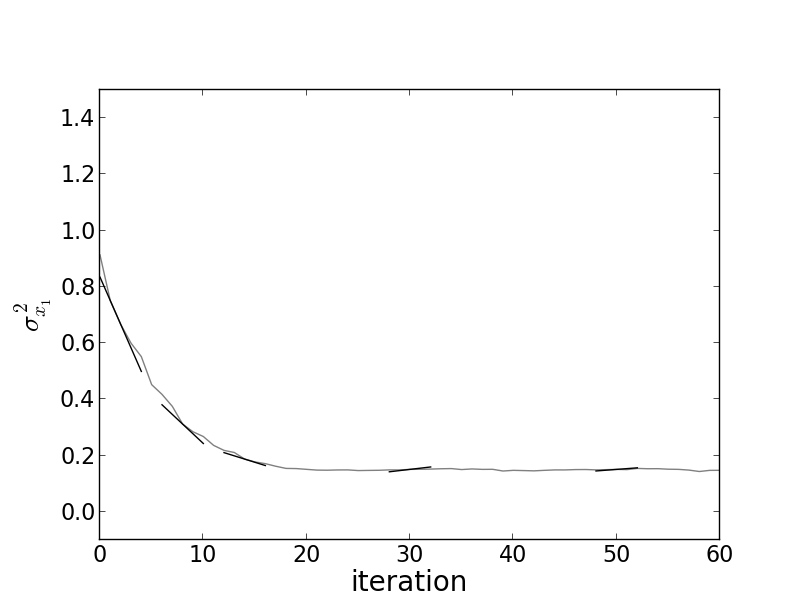}
\caption{$n=100$} 
\label{fig:5c}
\end{subfigure}  
\caption{In grey: the empirical variance $\Var_t(X_1)$, as estimated by the first coordinates $X_1^{(j)}(t)$ of all walkers at time $t$, for $0\leq t\leq 60$.
Overlaid in black: five predicted ``tangent'' lines for the curve, at times $t=2, 8, 14, 30, 50$, based on the slopes predicted in Equation \eqref{eq:19}.}
\label{fig:fig5}
\end{figure}

\section{Discussion}
 Even though the AIES has been used with success in the past on numerous occasions, we advise caution in high dimensional problems. 
The benchmark model we used to probe the problems of the AIES was a relatively simple model which already fails at $n=100$, and we expect problems to be even worse for more strongly correlated, or more complex, models. 

Unsurprisingly, other MCMC methods work more efficiently for the benchmark model we chose.
For instance, if the target distribution is interpreted as a posterior distribution with an uncorrelated Gaussian as the prior, then the elliptical slice sampler of \cite{MurAdaMac2010} gives good results.
If the target distribution is interpreted as a time-discretisation of a continuous process (in this case the Ornstein-Uhlenbeck process $U_t$, the centred Gaussian process with $\Cov(U_t,U_s)=e^{-\abs{t-s}}$, over the time interval $[0,\alpha n]$) then the ideas of \cite{CotRobStuWhi2013} can be used to obtain an MCMC method that handles dimensionality well.
However, these alternative methods require additional structure and analysis of the model.
Especially if the model is more complicated, the possibility of slow convergence may be a reasonable price to pay for the simplicity and generality of the AIES; the difficulty here is that the slow convergence can be hard to detect.

Indeed, an important issue in practice is how to know whether the AIES is experiencing the kind of problems we describe.
Evidently, it is not possible to rely on knowing the true target distribution, as we did.

Our analysis shows that the \emph{profile} of accepted $Z$ values of the correlated Gaussian-- particularly if they clustered on either side of $z=1$ -- gave relevant information about the system.
In both of the test cases the adapted Gelman-Rubin diagnostic gave a good  indication of lack of convergence. This method might fail is $W$ or both $W$ and $B$ are singular, which might indicate an ill posed problem, or very high correlation between parameters. \\
   Finally, traceplots are common tools to visually assess the performance of MCMC methods.
Because the AIES is an ensemble method, some adaptations are necessary.
For instance, the straightforward traceplots in Figure~\ref{fig:1a}, \ref{fig:1b} and \ref{fig:1c}, showing the first coordinates over all walkers and steps, give an impression of the range of likely values but give little insight into the different mixing properties.
Instead, selecting a small number of walkers, as in Figures~\ref{fig:1j}, \ref{fig:1k} and \ref{fig:1l}, shows that individual walkers are mixing well (relative to the range of likely values) when $n=10$, but not when $n=50$ or $n=100$.
In Figure~\ref{fig:1j} for $n=10$, it seems like the first coordinate of each walker is free to explore parameter space (between about $-2$ and 2, a region where most walkers appear to spend most of their time, according to Figure~\ref{fig:1a}). 
By contrast, in Figures~\ref{fig:1k} and \ref{fig:1l} for $n=50$ and $n=100$, the first coordinates of each walker seem to be confined to much narrower regions (even accounting for unduly restricted range of walker positions in Figures~\ref{fig:1b} and \ref{fig:1c}) and the relative order among the sampled walkers changes much less frequently.
In our examination of the AIES in this high-dimensional model, this was the only sign of slow convergence that we could identify without knowing the characteristics of the true target distribution.
\indent  
 The performance of the AIES can be improved by implementing the optimal choice for parameters $a$ and $L$. 
The stretch parameter $a$ can adjusted depending on the acceptance rate which should typically be between $0.2$ to $0.5$. The acceptance fraction can be increased if it's too low by decreasing $a$, and it can be decreased if it's to high be increase $a$. A large $L$ would also improve the performance \citep{ForHogLanGoo2013}.\\
   In summary, high dimensions can bring problems that make the AIES converge slowly and, more disturbingly, appear to have converged even when it has not.
  Knowing the structure of the true distribution allowed us to make accurate predictions about the evolution of the AIES for our chosen model. Looking at two measures arising from the algorithm -- the profile of accepted $Z$ values, a subsample traceplot of a few walkers, the adapted Gelman-Rubin-diagnostics and the adapted Heidelberger-Welch test-- gave a possible signal of slow convergence. Such diagnostics, and a measure of caution, should be used when applying the AIES to high-dimensional problems.

\section{Acknowledgments}
We would like to thank Michael Betancourt, Bob Carpenter, Andrew Gelman,
and Jeorg Dietrich for valuable discussion and comments. We also thank the
reviewers of an earlier version of this paper for their constructive criticisms.

\bibliographystyle{anzsj}
\bibliography{references}

\end{document}